\documentclass{llncs}

\newif\ifRR

\RRtrue

\usepackage{xspace}

\usepackage[utf8]{inputenc}
\usepackage{lmodern}
\usepackage{microtype}
\usepackage{enumerate}
\usepackage{thmtools}
\usepackage{thm-restate}
\usepackage{amsmath}
\usepackage{amssymb}
\usepackage{wrapfig}
\usepackage{subfigure}

\title{Degree of sequentiality of weighted automata~\thanks{This work has been funded by the DeLTA project (ANR-16-CE40-0007).}}

\author{Laure Daviaud\inst{1}\thanks{This author was partially supported by ANR Project ELICA ANR-14-CE25-0005, ANR Project RECRE ANR-11-BS02-0010 and by project \textsc{lipa} that has received funding from the European Research Council (\textsc{erc}) under the European Union’s Horizon 2020 research and innovation programme (grant agreement Nb 683080).} \and Isma{\"e}l Jecker\inst{2} \and Pierre-Alain Reynier\inst{3} \and Didier Villevalois\inst{3}}

\institute{
Warsaw University, Poland
\and
Université Libre de Bruxelles, Belgium
\and
Aix-Marseille Univ, LIF, CNRS, France
}

\usepackage{xcolor}

\usepackage{tikz}
\usetikzlibrary{arrows,shapes,automata, calc, chains, matrix, positioning, scopes}

\newcommand{\ie}{\emph{i.e.}\ }
\newcommand{\resp}{\emph{resp.}\ }
\newcommand{\intro}[1]{\textit{#1}}

\renewcommand{\leq}{\leqslant}
\renewcommand{\geq}{\geqslant}

\newcommand{\trans}[2]{\ensuremath{\textcolor{magenta}{#1}|\textcolor{blue}{#2}}}
\newcommand{\weight}[1]{\ensuremath{\textcolor{blue}{#1}}}

\renewcommand{\emptyset}{\varnothing}

\newcommand{\dom}{\textsf{dom}}
\newcommand{\N}{\mathbb{N}}
\newcommand{\alphabet}{A}
\newcommand{\semi}{\mathbb{S}}
\newcommand{\monoid}{\mathbb{M}}
\newcommand{\group}{\mathbb{G}}
\newcommand{\free}{\mathcal{F}(B)}

\newcommand{\unit}{\mathbf{1}}
\newcommand{\zero}{\mathbf{0}}

\newcommand{\partie}[1]{\mathcal{P}_\mathsf{fin}(#1)}
\newcommand{\up}{\mathcal{UP}}
\newcommand{\val}{\mathcal{V}}
\newcommand{\semipart}{\mathbb{P}_\mathsf{fin}}

\newcommand{\commentPA}[1]{}

\newcommand{\gene}{\Gamma}
\newcommand{\base}{M_{\poids}}

\newcommand{\dist}{\textsl{dist}}
\newcommand{\lcp}{\textsl{lcp}}
\newcommand{\delay}[2]{delay(#1,#2)}
\newcommand{\delaymot}[2]{\textsl{dist}(#1,#2)}
\newcommand{\cayley}[2]{d(#1,#2)}

\newcommand{\BTP}[1]{\ensuremath{\textsl{BTP}_{#1}}}

\newcommand{\lip}[1]{\ensuremath{\textsl{Lip}_{#1}}}

\newcommand{\Kc}{\ensuremath{L}}

\newcommand{\TP}[1]{\ensuremath{\textsl{TP}_{#1}}}

\newcommand{\poids}{W}
\newcommand{\reg}{C}

\newcommand{\alphab}{B}

\newcommand{\PTIME}{\ensuremath{\textsc{Ptime}}\xspace}
\newcommand{\PSPACE}{\ensuremath{\textsc{Pspace}}\xspace}

\newcommand{\Reg}{\mathcal{X}}

\newcommand{\inter}[1]{ \lbrack \! \lbrack #1 \rbrack \!\rbrack}

\newcommand{\rank}[1]{ \textsf{rank}(#1)}
\newcommand{\selec}[1]{ \textsf{select}(#1)}
\newcommand{\sele}{ \textsf{select}}
\newcommand{\select}[1][]{%
\ifthenelse{\equal{#1}{}}{ \textsf{select}_1}{ \textsf{select}_1(#1)}%
}
\newcommand{\selectt}[1][]{%
\ifthenelse{\equal{#1}{}}{ \textsf{select}_2}{ \textsf{select}_2(#1)}%
}
\newcommand{\prop}[1]{\textbf{\textsf{P}$_{#1}$}}

\newcommand{\inds}{\tau}
\newcommand{\indr}{\eta}

\newcommand{\tinit}{t_{\textsl{init}}}
\newcommand{\tfinal}{t_{\textsl{final}}}

\newcommand{\chunk}[1]{\ensuremath{(#1)\text{-chunk}}}

\newcommand{\ini}[3]{\ensuremath{#1 \in \{ #2, \ldots, #3 \}}}

\newcommand{\inc}{\ensuremath{\!\!+\!\!+}}

\newcommand{\flast}{\ensuremath{f_{last}}}

\usepackage{xifthen}

\tikzset{
  initial/.style={
    initial by arrow,
    initial where=left,
    initial distance=1cm,
    initial text=,
    append after command={%
      \pgfextra
        \pgfinterruptpath
          \node[above left, xshift=-0.2cm] at (\tikzlastnode.west) {#1};
        \endpgfinterruptpath
      \endpgfextra
    },
  },
  accepting/.style={
    accepting by arrow,
    accepting where=right,
    accepting distance=1cm,
    accepting text=,
    append after command={%
      \pgfextra
        \pgfinterruptpath
          \node[above right, xshift=+0.2cm] at (\tikzlastnode.east) {#1};
        \endpgfinterruptpath
      \endpgfextra
    },
  },
}

\makeatletter

\tikzstyle{accepting by relation}=    [after node path=
{
  {
    [to path=
    {
      [->,double=none,every accepting by relation]
      --
      ([shift=(\tikz@accepting@angle:\tikz@accepting@distance)]\tikztostart.\tikz@accepting@angle)
          node [shape=rectangle,anchor=\tikz@accepting@text@anchor,
            shift=(\tikz@accepting@angle:-\tikz@accepting@distance/2),
          ] {\tikz@accepting@text}
      }]
    edge ()
  }
}]
\tikzstyle{every accepting by relation}=[]

\tikzstyle{initial by relation}=   [after node path=
{
  {
    [to path=
    {
      [->,double=none,every initial by relation]
      ([shift=(\tikz@initial@angle:\tikz@initial@distance)]\tikztostart.\tikz@initial@angle)
          node [shape=rectangle,anchor=\tikz@initial@text@anchor,
            shift=(\tikz@initial@angle:-\tikz@initial@distance/2),
          ] {\tikz@initial@text}
        -- (\tikztostart)}]
    edge ()
  }
}]
\tikzstyle{every initial by relation}=[]

\tikzoption{initial text anchor}{\tikzaddafternodepathoption{\def\tikz@initial@text@anchor{#1}}}
\tikzoption{accepting text anchor}{\tikzaddafternodepathoption{\def\tikz@accepting@text@anchor{#1}}}

\def\tikz@initial@text@anchor{west}
\def\tikz@accepting@text@anchor{west}

\makeatother

\newcommand{\notePA}{}

\begin{document}

\maketitle


\begin{abstract}


Weighted automata (WA) are an important formalism to
describe quantitative properties. 
Obtaining equivalent deterministic machines is a longstanding 
research problem. 
In this paper we consider WA with a set semantics,
meaning that the semantics is given by the set of weights of accepting runs.
We focus on multi-sequential WA that are defined as finite
unions of sequential WA. 
The problem we address is to minimize the size of this union. We call this minimum the
degree of sequentiality of (the relation realized by) the WA.

For a given positive integer $k$, we provide multiple
characterizations of relations realized by a union of $k$
sequential \notePA WA over an infinitary finitely generated group: 
a Lipschitz-like machine independent property,
a pattern on the automaton (a new twinning property)
and a subclass of cost register
automata. When possible, we effectively translate a WA into an
equivalent union of $k$ sequential WA.
%
We also provide a decision procedure for our twinning property for commutative computable groups
thus allowing to compute the degree of sequentiality. Last, we show that these results also hold
for word transducers and that the associated decision problem is \PSPACE-complete.

\end{abstract}


\section{Introduction}
\label{sec:introduction}


%
%
%
%

\paragraph*{Weighted automata.}
Finite state automata can be viewed as functions from words to
Booleans and, thus, describe languages. Such automata have been
extended to define functions from words to various
structures yielding a very rich literature, with
recent applications in quantitative verification~\cite{CDH10}.
Weighted automata~\cite{Sch61} (WA) is the oldest of such formalisms.
They are defined over
semirings $(S,\oplus,\otimes)$ by adding weights from $S$ on
transitions;
the weight of a run is the product of the weights of the transitions,
and the weight of a word $w$ is the sum of the weights of the
accepting runs on $w$.

The decidability status of natural decision problems
such as universality and equivalence highly depends on the considered
semiring~\cite{Sakabook}. The first operation of the semiring,
used to aggregate the values computed by the different runs,
plays an important role in the (un)decidability results.
Inspired by the setting of word transducers, recent works have
considered a set semantics that consists in keeping all these
values as a set, instead of aggregating them~\cite{Manu-LMCS},
and proved several decidability results for the resulting class
of finite-valued weighted automata~\cite{DBLP:conf/fsttcs/FiliotGR14}.

For automata based models, a very important problem is to simplify
the models. For instance,
deterministic (a.k.a. sequential) machines allow to derive efficient
evaluation algorithms.
In general, not every WA can be transformed into an equivalent
sequential one. The sequentiality problem
then asks, given a WA on some semiring $(S,\oplus,\otimes)$,
whether there exists an equivalent sequential WA over $(S,\oplus,\otimes)$.
This problem ranges from trivial to undecidable, depending on
the considered semiring, see~\cite{LS06} for a survey and~\cite{KirstenL09,Kirsten12} for more recent works.

%
%
%
%
%
%
%
%

\paragraph*{Sequential transducers.}

Transducers define rational relations over words. They
can be viewed as weighted automata over the semiring of finite sets of
words (thus, built over the free monoid); sum is the set union and
product is the concatenation extended to sets.
 When the underlying
automaton is deterministic, then the transducer
is said to be \emph{sequential}. The class of sequential functions, \emph{i.e.}
those realized by sequential transducers, has been characterized among
the class of rational functions by Choffrut, see for instance~\cite{BealC02} for a
presentation:
\begin{theorem}[\cite{Choffrut77}]
\label{thm:Choffrut}
Let $T$ be a functional finite state transducer and $\inter{T}$ be the function realized by $T$.
The following assertions are equivalent:
\begin{enumerate}[$i)$]
\item $\inter{T}$ satisfies the bounded variation property
\item $T$ satisfies the twinning property
\item $\inter{T}$ is computed by a sequential transducer
\end{enumerate}
\end{theorem}
In this result, two key tools are introduced: a property of the function, known
as the \emph{bounded variation property}, and a pattern property of the transducer,
known as the \emph{twinning property}.

\paragraph*{Multi-sequential weighted automata.}

Multi-sequential functions
of finite words have been introduced in~\cite{ChoffrutS86}
as those functions that can be realized by a finite
union of sequential transducers. A characterization of these
functions among the class of rational functions is given in~\cite{ChoffrutS86}.
Recently, this definition has been lifted to relations
in~\cite{JeckerF15} where it is proved that the class of so-called
multi-sequential relations can be decided in \PTIME
among the class of rational relations.

We consider in this paper \emph{multi-sequential weighted automata},
defined as finite unions of sequential WA. As described above,
and following~\cite{Manu-LMCS}, we consider weighted automata
with a set semantics.
We argue that multi-sequential WA are an interesting
compromise between sequential and non-deterministic ones.
Indeed, sequential WA have a very low expressiveness, while
it is in general difficult to have efficient evaluation procedures
for non-deterministic WA. Multi-sequential WA allow
to encode standard examples requiring non-determinism,
yet provide a natural evaluation procedure.
Multi-sequential WA can indeed be efficiently 
evaluated in parallel by using a thread per member of the union, thus
avoiding inter-thread communication.

A natural problem consists in minimizing the size
of the union of multi-sequential WA that is,
given a WA and a natural number $k$,
decide whether it can be realized as a union of $k$
sequential WA. We are also interested
in identifying the minimal such $k$, that we call \emph{degree of sequentiality}
of the WA.

\paragraph*{Contributions.}

In this paper, we propose a solution to the problem of the computation of
the degree of sequentiality of WA. Following previous works~\cite{Manu-LMCS,daviaud_2016}, we consider
WA over infinitary finitely generated groups. We introduce new generalizations of
the tools of Choffrut that allow us to characterize the relations that can be
defined as unions of $k$ sequential WA: first, a property of relations
that extends a Lipschitz property for transducers, and is called
\emph{Lipschitz property of order~$k$} (\lip{k} for short); second, a
pattern property of transducers, called \emph{branching
twinning property of order~$k$} (\BTP{k} for short). We prove:
\begin{theorem}
\label{theorem:main}
Let $W$ be a weighted automaton with set semantics over an infinitary finitely generated group
and $k$ be a positive integer.  The following
assertions are equivalent:
\begin{enumerate}[$i)$]
\item $\inter{W}$ satisfies the Lipschitz property of order~$k$, \label{th-lip}
\item $W$ satisfies the branching twinning property of order~$k$,\label{th-btp}
\item $\inter{W}$ is computed by a $k$-sequential weighted automaton, \label{th-wa}
\end{enumerate}
In addition, the equivalent model of property \ref{th-wa}) 
can be effectively computed.
\end{theorem}
As demonstrated by this result, the first important contribution of our work is thus to
identify the correct adaptation of the properties of Choffrut suitable to characterize
$k$-sequential relations. Sequential functions are characterized
by both a bounded variation and a Lipschitz property~\cite{berstel_transductions_2013}.
In~\cite{daviaud_2016}, we introduced a generalization
of the bounded variation property to characterize relations that can be
expressed using a particular class of cost register automata with exactly $k$ registers, that encompasses
the class of $k$-sequential relations.
Though, to characterize $k$-sequential relations,
we here introduce
a generalization of the Lipschitz property. We actually believe that this class
cannot be characterized by means of a generalization of the bounded variation
property.
Similarly, the difference between the twinning property of order~$k$ introduced
in~\cite{daviaud_2016} and the branching twinning property of order~$k$
introduced in this paper is subtle: we allow here to consider runs on different input
words, and the property requires the existence of two runs whose outputs are close
on their common input words.


We now discuss  the proof of Theorem~\ref{theorem:main}
whose structure is depicted in the picture on the right.
In~\cite{Choffrut77}, as well as in~\cite{daviaud_2016},
the difficult part is the construction, given
a machine satisfying the pattern property, of an equivalent deterministic
machine.
Here again,
the most intricate proof of our work is that of Proposition~\ref{prop:BTP-kseq}:
the construction, given a WA satisfying the $\BTP{k}$, of an equivalent
$k$-sequential weighted automaton.
It is worth noting that it is not a simple extension of~\cite{Choffrut77} and~\cite{daviaud_2016}.
\begin{wrapfigure}{r}{0.4\textwidth}
 \centering
\ifRR
\includegraphics[width=5cm]{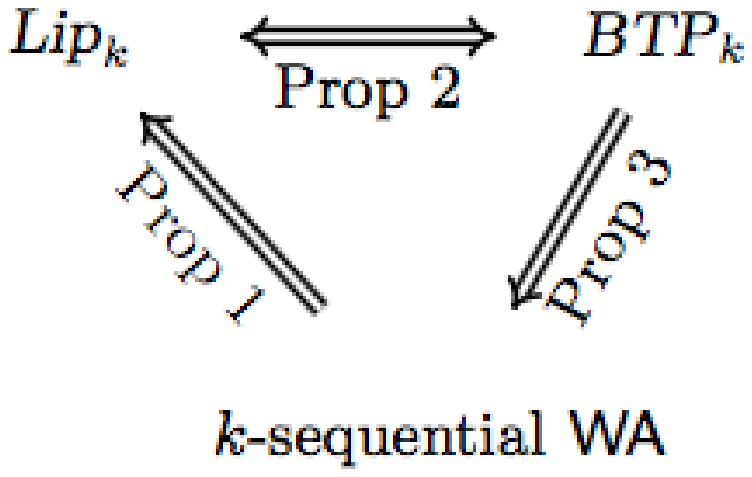}
\fi
\end{wrapfigure}
Our proof proceeds by induction on $k$, and the result of~\cite{Choffrut77} constitutes
the base case while the tricky part resides in the induction step.
Compared with~\cite{daviaud_2016}, the construction of~\cite{daviaud_2016} stores pairwise
delays between runs,
and picks a minimal subset of "witness" runs that allows to express every other run.
In~\cite{daviaud_2016}, the choice of these witnesses may evolve along an execution
while in order to define a $k$-sequential WA,
the way we choose the representative runs should be consistent during the execution.
The technical part of our construction is thus the identification of a partition of size at most $k$ of
the different runs of the non-deterministic WA such that each element of this partition
defines a sequential function.
 This relies on the branching
structure of the twinning property we introduce in this paper.

%

Our result can also be rephrased in terms of cost register automata~\cite{DBLP:conf/lics/AlurDDRY13}.
These
are deterministic automata equipped with registers that aim to store along
the run values from a given semiring $S$.
The restriction of this model to updates of the form $X:=X\alpha$ (we say that registers are \emph{independent})
exactly coincides (if we allow $k$ registers) with the class of $k$-sequential relations.
Hence, our result also allows to solve the register
minimization problem for this class of CRA.

Beyond weighted automata over infinitary groups, we also prove that
our results apply to transducers from $A^*$ to $B^*$.

Regarding decidability, we show that if the group $\group$
is commutative and has a computable internal operation, then
checking whether the \BTP{k} is satisfied is decidable.
As a particular instance of our decision procedure, we
obtain that this can be decided in \PSPACE for $\group=(\mathbb{Z},+,0)$,
and show that the problem is \PSPACE-hard.
Last, we prove that checking the \BTP{k}
for finite-state transducers
is also \PSPACE-complete.

\paragraph*{Organization of the paper.}
We start with definitions in Section~\ref{sec:preliminaries}. In
Section~\ref{sec:tp}, we introduce our original Lipschitz and branching twinning properties.
We present our main construction in 
Section~\ref{sec:result}.
Section~\ref{sec:cra} is devoted to the presentation
of our results about cost register automata,
while transducers are dealt with
in Section~\ref{sec:transducers}. Last
we present our decidability results and their application to the
computation of the degree of sequentiality in Section~\ref{sec:decision}.
Omitted proofs can be found in the Appendix.

\section{Definitions and examples}
\label{sec:preliminaries}

\paragraph*{Prerequisites and notation.}
We denote by $\alphabet$ a finite alphabet, by $\alphabet^*$ the set of finite
words on $\alphabet$, by $\varepsilon$ the empty word and by $|w|$ the length of a word $w$.
For a set $S$, we denote by $|S|$ the cardinality of $S$.

A \intro{monoid} $\monoid=(M,\otimes,\unit)$ is a set $M$ equipped
with an associative binary operation $\otimes$ with $\unit$ as neutral
element; the product $\alpha\otimes \beta$ in $M$ may be simply
denoted by $\alpha \beta$.
If every element of a monoid possesses an
inverse - for all $\alpha\in M$, there exists $\beta$ such that
$\alpha\beta=\beta\alpha=\unit$ (such a $\beta$ is unique and is
denoted by $\alpha^{-1}$) - then $M$ is called a group. The monoid (\textit{resp.} group)
is said to be \intro{commutative} when $\otimes$ is commutative.
Given a finite alphabet $\alphab$, we denote by $\free$ the free group
generated by $\alphab$.

A \intro{semiring} $\semi$ is a set $S$ equipped with two binary
operations $\oplus$ (sum) and $\otimes$ (product) such that
$(S,\oplus,\zero)$ is a commutative monoid with neutral element $\zero$,
$(S,\otimes,\unit)$ is a monoid with neutral element $\unit$, $\zero$ is
absorbing for $\otimes$ (\textit{i.e.}
$\alpha\otimes \zero=\zero\otimes \alpha=\zero$) and $\otimes$
distributes over $\oplus$ (\textit{i.e.} $\alpha\otimes
(\beta \oplus \gamma) = (\alpha\otimes \beta) \oplus
(\alpha\otimes \gamma)$ and $(\alpha \oplus \beta)\otimes \gamma =
(\alpha \otimes \gamma)\oplus (\beta \otimes \gamma))$.

Given a set $S$, the set of the finite subsets of $S$ is denoted by
$\partie{S}$.  For a monoid $\monoid$, the set $\partie{M}$ equipped
with the two operations $\cup$ (union of two sets) and the set
extension of $\otimes$ is a semiring denoted $\semipart(\monoid)$.

From now on, we may identify algebraic structures (monoid, group,
semiring) with the set they are defined on when the operations are
clear from the context.

\paragraph*{Delay and infinitary group.}
There exists a classical notion of \intro{distance} on words (\ie on
the free monoid) measuring their difference: $\textsl{dist}$ is
defined for any two words $u,v$ as $\delaymot{u}{v}=|u|+|v|-2*
|\lcp(u,v)|$ where $\lcp(u,v)$ is the longest common prefix of $u$ and
$v$.

When considering a group $\group$ and $\alpha,\ \beta \in \group$,
we define the \intro{delay} between $\alpha$ and $\beta$ as
$\alpha^{-1}\beta$, denoted by $\delay{\alpha}{\beta}$.

\begin{lemma}
\label{lemma:properties-delay}
Given a group $\group$, for all $\alpha,\alpha',\beta,\beta',\gamma,\gamma' \in \group$,
\vspace{-.15cm}
\begin{enumerate}
\item $\delay{\alpha}{\beta}=\unit$ if and only if $\alpha=\beta$,
\item if $\delay{\alpha}{\alpha'}=\delay{\beta}{\beta'}$ then $\delay{\alpha\gamma}{\alpha'\gamma'}= \delay{\beta\gamma}{\beta'\gamma'}$.
\end{enumerate}
\end{lemma}

For a finitely generated group $\group$, with a fixed finite set of
generators $\gene$, one can define a distance between two elements
derived from the Cayley graph of $(\group,\gene)$. We consider here an
undirected right Cayley graph : given $\alpha \in \group$,
$\beta \in \gene$, there is a (non-oriented) edge between $\alpha$ and
$\alpha \beta$.
Given $\alpha,\ \beta \in \group$, the \intro{Cayley
distance} between $\alpha$ and $\beta$ is the length of the shortest
path linking $\alpha$ and $\beta$ in the undirected right Cayley graph
of $(\group, \gene)$. It is denoted by $\cayley{\alpha}{\beta}$.

For any  $\alpha \in \group$, we define the \intro{size} of $\alpha$ (with respect to the set of generators $\gene$)
as the natural number $\cayley{\unit}{\alpha}$. It is denoted by $|\alpha|$. Note that for a word $u$,
considered as an element of $\mathcal{F}(A)$,
the size of $u$ is exactly the length of $u$ (that is why we use the same notation).

\begin{lemma}
Given a finitely generated group $\group$ and a finite set of
generators $\gene$, for all $\alpha, \beta \in \group$,
$\cayley{\alpha}{\beta} = |\delay{\alpha}{\beta}|$.
\end{lemma}

A group $\group$ is said to be \intro{infinitary} if for all
$\alpha,\beta,\gamma \in \group$ such that
$\alpha\beta\gamma \neq \beta $, the set
$\{\alpha^n \beta \gamma^n \mid n\in \N \}$ is infinite.
Classical examples of infinite groups such as $(\mathbb{Z},+,0)$,
$(\mathbb{Q},\times,1)$ and the free group generated by a finite
alphabet are all infinitary. See~\cite{Manu-LMCS} for other examples.

\paragraph*{Weighted automata.}
Given a semiring $\semi$, weighted automata (WA) are non-deterministic
finite automata in which transitions have for weights elements of
$\semi$. Weighted automata compute functions from the set of words to $\semi$: the
weight of a run is the product of the weights of the transitions along
the run and the weight of a word $w$ is the sum of the weights of the
accepting runs labeled by $w$.

We will consider, for some monoid $\monoid$, weighted automata over
the semiring $\semipart(\monoid)$. In our settings, instead of considering the
semantics of these automata in terms of functions from $\alphabet^*$ to
$\semipart(\monoid)$, we will consider it in terms of relations over
$\alphabet^*$ and $\monoid$.
More precisely, a weighted automaton
(with initial and final relations),
is formally defined as follows:

\begin{definition}
Let $\alphabet$ be a finite alphabet, a \intro{weighted automaton} $W$
over some monoid $\monoid$
is a tuple $(Q, t_{\textsl{init}},
t_{\textsl{final}}, T)$ where $Q$ is a finite set of states,
$t_{\textsl{init}} \subseteq Q \times \monoid$ (resp. $t_{\textsl{final}} \subseteq Q \times \monoid$)
is the finite initial (resp. final) relation,
$T \subseteq Q \times \alphabet \times \monoid \times Q$ is the finite
set of transitions.
\end{definition}

A state $q$ is said to be \intro{initial} (\resp \intro{final})
  if there is $\alpha \in \monoid$ such that
    $(q,\alpha) \in t_{\textsl{init}}$ (\resp $(q,\alpha) \in t_{\textsl{final}}$),
  depicted as $\xrightarrow{\alpha} q$ (resp. $q \xrightarrow{\alpha}$).
A run $\rho$ from a state $q_1$ to a state $q_k$ on a word $w = w_1 \dotsm w_k \in \alphabet^*$
  where for all $i$, $w_i \in \alphabet$, is a sequence of transitions:
  $(q_1,w_1,\alpha_1,q_2),(q_2,w_2,\alpha_2,q_3),\ldots,(q_k,w_k,\alpha_k,q_{k+1})$.
The \intro{output} of such a run is the element of $\monoid$, $\alpha = \alpha_1\alpha_2 \dotsm \alpha_k$.
We depict this situation as $q_1 \xrightarrow{ w\mid\alpha} q_{k+1}$.
The run $\rho$ is said to be \intro{accepting} if $q_1$ is initial and $q_{k+1}$ final.
This automaton $W$ computes a relation $\inter{W} \subseteq \alphabet^* \times \monoid$
defined by the set of pairs $(w,\alpha \beta \gamma)$ such that there are $p,q \in Q$
with $\xrightarrow{\alpha} p \xrightarrow{w \mid \beta} q \xrightarrow{\gamma}$.

An automaton is \intro{trimmed} if each of its states appears in some accepting run.
W.l.o.g., we assume that the automata we consider are trimmed.

Given a weighted automaton $W=(Q, t_{\textsl{init}},
t_{\textsl{final}}, T)$ over some finitely generated group $\group$
with finite set of generators $\Gamma$, we define the constant $\base$
with respect to $\Gamma$ as
$\base = \max\{|\alpha| \mid (p,a,\alpha,q)\in T \text{ or }
(q,\alpha) \in t_{\textsl{init}} \cup t_{\textsl{final}}\}$.

For any positive integer $\ell$, a relation $R \subseteq X \times Y$ is said to be \intro{$\ell$-valued}
  if, for all $x \in X$, the set $\{y \mid (x,y) \in R\}$ contains at most $\ell$ elements.
  It is said to be \intro{finitely valued} if it is $\ell$-valued for some $\ell$.
A weighted automaton $W$ is said to be \intro{$\ell$-valued} (resp. \intro{finite-valued})
  if it computes a $\ell$-valued (resp. finite-valued) relation.

The union of two weighted automata $W_i = (Q_i, t_{init}^i, t_{final}^i, T_i)$, for $i \in \{1,2\}$,
with disjoint states $Q_1 \cap Q_2 = \emptyset$ is the automaton $W_1 \cup W_2 =
(Q_1 \cup Q_2, t_{init}^1 \cup  t_{init}^2, t_{final}^1 \cup  t_{final}^2, T_1 \cup  T_2)$.
States can always be renamed to ensure disjointness. It is trivial to
verify that $\inter{W_1\cup W_2}=\inter{W_1}\cup\inter{W_2}$.
This operation can be generalized
to the union of $k$ weighted automata.

\begin{definition}
A weighted automaton $(Q, t_{\textsl{init}},
t_{\textsl{final}}, T)$ over $\monoid$ is said to be
\intro{sequential} if $|t_{\textsl{init}}| = 1$ and if for all $p\in Q$, $a\in \alphabet$ there is at most one transition in $T$ of the form $(p,a,\alpha,q)$.
It is said to be \intro{$k$-sequential} if it is a union of $k$ sequential automata.
It is said to be \intro{multi-sequential} if it is $k$-sequential for some $k$.
A relation is said to be \intro{$k$-sequential} (resp. \intro{multi-sequential}) \notePA
  if it can be computed by a $k$-sequential (resp. multi-sequential) automaton.
The \emph{degree of sequentiality} of the relation is the minimal $k$ such that it is $k$-sequential.
\end{definition}


Observe that,
unlike the standard definition of sequential weighted automata over $\monoid$
  (see for instance~\cite{Manu-LMCS}),
we allow finite sets of weights to be associated with final states,
  and not only singletons.
This seems more appropriate to us regarding the parallel evaluation model
for multi-sequential weighted automata:
we prefer to merge threads that only differ by their final outputs.
%
If we define $OutMax = \max_{q \in Q} |\{ (q,\alpha) \in t_{\textsl{final}} \}|$,
  then the standard definition of sequential machines requires $OutMax = 1$.
Being $k$-sequential implies being $(k \cdot OutMax)$-valued.
Hence, multi-sequential weighted automata are included in finite-valued ones.
However, multi-sequential weighted automata are strictly less expressive than finite-valued ones.

Allowing a final output relation obviously has an impact on the sequentiality degree.
%
%
We believe that it is possible to fit the usual setting
by appropriately reformulating our characterizations.
However, this cannot be directly deduced from our current results.

\begin{example}
\label{example:weighted}
Let us consider $\alphabet =\{a,b\}$ and $(\monoid,\otimes,\unit)=(\mathbb{Z},+,0)$.
The weighted automaton $W_0$ given in Figure~\ref{example-wa}
  computes the function $\flast$ that associates with a word $wa$ (resp. $wb$)
  its number of occurrences of the letter $a$ (resp. $b$),
  and associates $0$ with the empty word.
  It is easy to verify that the degree of sequentiality of $\flast$ is $2$.
  \notePA
It is also standard that the function $\flast^*$ mapping
  the word $u_1\#\ldots\#u_n$ (for any $n$)
  to $\flast(u_1) + \cdots + \flast(u_n)$
  is not multi-sequential (see for instance~\cite{JeckerF15})
  whereas it is single-valued.
\end{example}

\ifRR
  \begin{figure}[htb!]
    \centering
    \scalebox{.95}{
%
\begin{tikzpicture}[
  ->,
  >=stealth',
  shorten >=1pt,
  node distance=1cm,
  scale=0.8,
  every state/.style={minimum width=0.7cm,draw=blue!50,very thick,fill=blue!20, scale=0.8},
  every node/.style={font=\scriptsize},
  initial/.style={
    initial by relation,
    initial where=above,
    initial distance=.7cm,
    initial text=#1,
    initial text anchor=base west,
  },
  accepting/.style={
    accepting by relation,
    accepting where=below,
    accepting distance=.7cm,
    accepting text=#1,
    accepting text anchor=west,
  },
]
  \node[state, initial=$0$] (a) [] {$q_a$};
  \node[state, initial=$0$, accepting=$0$] (f) [right=of a] {$q_f$};
  \node[state, initial=$0$] (b) [right=of f] {$q_b$};

  \path[->]     (a) edge [loop left] node         {$b: 0$}   ();
  \path[->]     (a) edge [loop below] node         {$a: 1$}   ();
  \path[->]     (a) edge []  node [above, sloped] {$a: 1$}   (f);

  \path[->]     (b) edge [loop right] node         {$a: 0$}   ();
  \path[->]     (b) edge [loop below] node         {$b: 1$}   ();
  \path[->]     (b) edge []  node [above, sloped] {$b: 1$}   (f);
\end{tikzpicture}}
    \caption{Example of a weighted automaton $W_0$ computing the function $\flast$.
    \label{example-wa}}
  \end{figure}
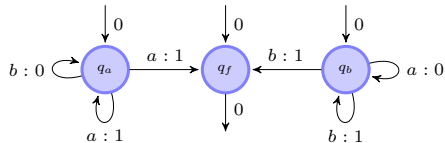
\fi

\section{Lipschitz and branching twinning properties}
\label{sec:tp}


\newcommand{\lipschitzCitation}{}

Sequential transducers have been characterized in~\cite{Choffrut77} by Choffrut by means
of a so-called bounded-variation property and a twinning property.
The bounded-variation property is actually equivalent to a Lipschitz-like property
(see for instance~\cite{berstel_transductions_2013}).
We provide adaptations of the Lipschitz and twinning properties
so as to characterize $k$-sequential WA.


We consider a finitely generated infinitary group $\group$
and we fix a finite set of generators $\Gamma$.

\subsection{Lipschitz property of order~$k$}

Given a partial mapping $f : A^* \rightharpoonup B^*$,
  the Lipschitz property states that there exists $\Kc \in \N$
  such that for all $w,w' \in A^*$ such that $f(w),f(w')$ are defined, \notePA
  we have $\dist(f(w),f(w')) \leq \Kc \, \dist(w,w')$ (see \cite{berstel_transductions_2013}).
Intuitively, this property states that, for two words, their images by $f$ differ proportionally to those words.
This corresponds to the intuition that the function can be expressed by means of a sequential automaton.

When lifting this property to functions that can be expressed using a $k$-sequential automaton,
  we consider $k+1$ input words
  and require that two of those must have proportionally close images by $f$.
The extension to relations $R \subseteq A^* \times B^*$ requires that
  for all $k+1$ pairs chosen in $R$,
  two of those have their range components proportionally close to their domain components. \notePA
  In addition, for relations, an input word may have more than one output word,
  we thus need to add a constant $1$ in the right-hand side.
Finally, our framework is that of infinitary finitely generated groups.
Instead of $\dist(,)$, we use the Cayley distance $\cayley{}{}$
  to compare elements in the range of the relation.


\begin{definition}
  A relation $R \subseteq \alphabet^* \times \group$
   satisfies the \intro{Lipschitz property of order~$k$}
  if  there is a natural $\Kc$ such that for all pairs $(w_0,\alpha_0), \ldots,
  (w_k,\alpha_k) \in R$,
  there are two indices $i, j$ such that $0\leq i < j \leq k$ and
  $\cayley{\alpha_i}{\alpha_j} \leq \Kc \, (\dist(w_i,w_j) + 1)$.
\end{definition}

\begin{example}
  The group $(\mathbb{Z},+,0)$ is finitely generated with $\{1\}$ as a set of generators. \notePA
  The function $\flast$ 
    does not satisfy the Lipschitz property of order $1$ (take $w_1=a^Na$ and $w_2=a^Nb$), but it
  satisfies the Lipschitz property of order $2$.
\end{example}

Using the pigeon hole principle, it is easy to prove the implication from
$\ref{th-wa})$ to $\ref{th-lip})$ of Theorem~\ref{theorem:main}:
\begin{proposition}\label{prop:kseq-lip}
  A $k$-sequential relation satisfies the Lipschitz property of order~$k$.
\end{proposition}

\subsection{Branching twinning property of order~$k$}

The idea behind the branching twinning property of order~$k$ is
  to consider $k+1$ runs labeled by arbitrary words
  with $k$ cycles.
If the branching twinning property is satisfied
  then there are two runs among these $k+1$
  such that the values remain close
  (i.e. the Cayley distance between these values is bounded)
  along the prefix part of these two runs that read the same input.
%
%
%
This property is named after the intuition
  that the $k+1$ runs can be organized in a tree structure
  where the prefixes of any two runs are on the same branch
  up to the point where those two runs do not read the same input anymore.

\begin{definition}
  \label{def:tp} A weighted automaton over $\group$ satisfies
  the \intro{branching twinning property of order~$k$} (denoted by \BTP{k}) if: (see Figure~\ref{figure:BTPk})
\begin{itemize}
\item for all states $\{q_{i,j} \mid i,j \in \{0,\dots, k\}\}$ with $q_{0,j}$ initial for all $j$,
\item for all $\gamma_j$ such that $(q_{0,j}, \gamma_j) \in t_{\textsl{init}}$ with $j \in \{0,\dots, k\}$,
\item for all words $u_{i,j}$ and $v_{i,j}$ with $1\leq i \leq k$ and $0\leq j \leq k$ such that there are $k+1$ runs satisfying for all $0 \leq j \leq k$, for all $1 \leq i \leq k$, $q_{i-1,j} \xrightarrow{u_{i,j} \mid \alpha_{i,j}} q_{i,j}$ and $q_{i,j} \xrightarrow{v_{i,j} \mid \beta_{i,j}} q_{i,j}$,
\end{itemize}
there are $j \neq j'$ such that for all $i \in \{1,\ldots, k\}$, if for every $1\leq i' \leq i$,
we have $u_{i',j}=u_{i',j'}$ and $v_{i',j}=v_{i',j'}$, then we have
\begin{align*}
\delay{\gamma_j \alpha_{1,j} \dotsm \alpha_{i,j}}{\gamma_{j'} \alpha_{1,j'} \dotsm \alpha_{i,j'}} = \delay{\gamma_j \alpha_{1,j} \dotsm \alpha_{i,j} \beta_{i,j}}{\gamma_{j'} \alpha_{1,j'}  \dotsm \alpha_{i,j'} \beta_{i,j'}}.
\end{align*}
\end{definition}

\begin{example}\label{example-btp}
  The weighted automaton $W_0$, given in Figure~\ref{example-wa},
    does not satisfy the $\BTP{1}$ (considering loops around $q_a$ and $q_b$). \notePA
  One can prove however that it satisfies the $\BTP{2}$.

Let us denote by $W_1$ the weighted automaton obtained by concatenating
$W_0$ with itself, with a fresh $\#$ separator letter. $W_1$
realizes the function $\flast^2$ defined as
 $\flast^2(u\#v)=\flast(u)+\flast(v)$. We can see that the minimal $k$ such that $W_1$ satisfies
 the $\BTP{k}$ is $k=4$. As we will see, this is the sequentiality degree of $\flast^2$.
%
\end{example}

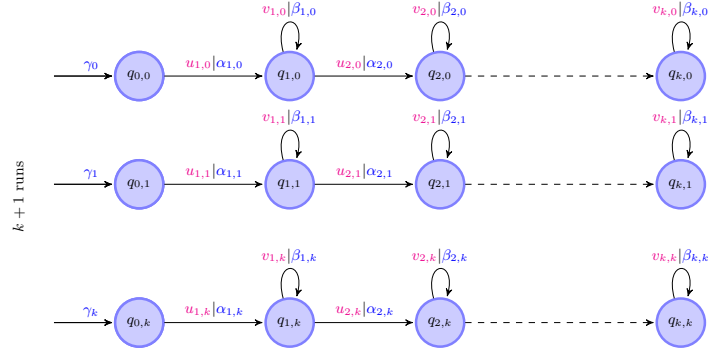
\begin{figure}
  \centering
  \def \loopCount {k}
\scalebox{.8}{
\begin{tikzpicture}[
  ->,
  >=stealth',
  shorten >=1pt,
  auto,
  node distance=2.8cm,
  every state/.style={minimum width=1cm,draw=blue!50,very thick,fill=blue!20},
  every node/.style={font=\small,scale=0.8},
]
  \tikzstyle{every edge}=[draw=black,font=\small]

  \foreach \j/\y in {0/0,1/-1.8,k/-4.1}
  {
    \node[initial=$\weight{\gamma_\j}$,state] (A0) at (0,\y) {$q_{0,\j}$};

    \foreach \i/\x in {1/2.5,2/5,\loopCount/9} {
      \node[state] (A\i)  at (\x,\y) {$q_{\i,\j}$};

      \path (A\i) edge [loop above] node {\trans{v_{\i,\j}}{\beta_{\i,\j}}} (A\i);
    }

    \foreach \im/\i in {0/1,1/2} {
      \path (A\im) edge node {\trans{u_{\i,\j}}{\alpha_{\i,\j}}} (A\i);
    }
    \foreach \im/\i in {2/\loopCount} {
      \path (A\im) edge [dashed] node {} (A\i);
    }
  }

  \node[rotate=90] (textleft) at (-2,-2) {$k+1$ runs};
\end{tikzpicture}
}
  \caption{Branching twinning property of order~$k$}
  \label{figure:BTPk}
\end{figure}

\subsection{Equivalence of Lipschitz and branching twinning properties}


We can prove that a weighted automaton satisfies the $\BTP{k}$ if and only if its semantics satisfies the Lipschitz property of order~$k$. This implies that the branching twinning property of order~$k$ is a machine independent property, \emph{i.e.} given two WA $W_1,W_2$
such that $\inter{W_1}=\inter{W_2}$, $W_1$ satisfies the $\BTP{k}$
iff $W_2$ satisfies the $\BTP{k}$.

\begin{proposition}\label{prop:BTP-lip-equivalence}
A weighted automaton $\poids$ over an infinitary finitely generated group $\group$ satisfies $\BTP{k}$ if and only if $\inter{W}$ satisfies the Lipschitz property of order~$k$.
\end{proposition}

\begin{proof}[Sketch]
Let us sketch the proof of the Proposition. 
First, suppose that $\poids$ does not satisfy the $\BTP{k}$. Then consider a witness of this non satisfaction. Fix an integer $\Kc$. By pumping the loops in this witness (enough time and going backward), one can construct $k+1$ words that remain pairwise sufficiently close while their outputs are pairwise at least at distance $\Kc$. This leads to prove that $\inter{\poids}$ does not satisfy the Lipschitz property of order~$k$.

Conversely, consider that the $\BTP{k}$ is satisfied.
For all $k+1$ pairs of words and weights in $\inter{\poids}$,
  we have $k+1$ corresponding runs in $\poids$ labeled by those words.
By exhibiting cycles on these runs, we can get an instance of $\BTP{k}$ as in Figure \ref{figure:BTPk} such that the non-cycling part is bounded (in length).
By $\BTP{k}$, there are two runs that have the same delays before and after the loops appearing in their common prefix. Thus, we can bound the distance between the
two weights produced by those runs
proportionally to the distance between the two input words, proving that the Lipschitz property is satisfied. \qed
\end{proof}


\section{Constructing a $k$-sequential weighted automaton}
\label{sec:result}

As explained in the introduction, the most intricate part in the proof of Theorem~\ref{theorem:main} is to
prove that $\ref{th-btp})$ implies $\ref{th-wa})$. We give a constructive proof of this fact as
stated in the following proposition.

\begin{proposition}\label{prop:BTP-kseq}
Given a weighted automaton $W$ satisfying the $\BTP{k}$, one can effectively build $k$ sequential weighted automata whose union is equivalent to $W$.
\end{proposition}

Let $W = (Q,t_{\textsl{init}},t_{\textsl{final}}, T)$ be a weighted automaton that satisfies the \BTP{k}.
The construction is done in two steps.
First, we build an infinite sequential weighted automaton $D_W$ equivalent to $W$, using the subset construction with delays presented in \cite{BealC02}.
Then, by replacing infinite parts of $D_W$ with finite automata, we build $k$ sequential weighted automata whose union is equivalent to $W$.

Let us sketch the main ideas behind the construction of $D_W$.
The states of $D_W$ are the subsets $S$ of $Q \times \group$.
On input $u \in A^*$, $D_W$ selects an initial run
$
\rho : \xrightarrow{\alpha_0} p_0 \xrightarrow{u|\alpha} p
$
of $W$, outputs the corresponding $\alpha \in \group$, and, in order to keep track of all the runs 
$
\rho' : \xrightarrow{\beta_0} q_0 \xrightarrow{u|\beta} q
$
 of $W$ over the input $u$, stores in its state the corresponding pairs $(q',\delay{\alpha_0 \alpha}{\beta_0 \beta})$.
%
The detailed construction, together with the proofs of its properties, adapted from \cite{BealC02} to fit our settings, can be found in the appendix.

If $W$ is a transducer, i.e., a weighted automaton with weights in a free monoid, and $W$ satisfies the $\BTP{1}$, which is equivalent to the twinning property, Lemma 17 of \cite{BealC02} proves that the trim part of $D_W$ is finite.
This lemma can be generalized to any kind of weighted automata, proving our proposition in the particular case $k = 1$.
Let us now prove the general result by induction.
Suppose that $k > 1$, and that the proposition is true for every integer strictly smaller than $k$.
We begin by exposing two properties satisfied by $D_W$.

Since $W$ satisfies the \BTP{k}, it also satisfies the notion of $\TP{k}$ introduced in \cite{daviaud_2016}, and, by Proposition 1 of that paper, it is $\ell$-valued for some integer $\ell$ effectively computable.
Let $N_W = 2M_W|Q|^{\ell |Q|}$, let $S \in Q \times \group$ be a state of the trim part of $D_W$, and let $W_S = (Q,S,t_{\textsl{final}}, T)$ be the weighted automaton obtained by replacing the initial output relation of $W$ with $S$. The following properties are satisfied.
\begin{description}
\item[\prop{1}:]
The size of $S$ is bounded by $\ell |Q|$;
\end{description}
\begin{description}
\item[\prop{2}:]
If there exists a pair $(q, \alpha) \in S$ such that $|\alpha|>N_W$, $\inter{W_S}$ is $k$-sequential.
\end{description}
The proof of \prop{1} follows from the $\ell$-valuedness of $W$.
The main difficulty of the demonstration of Proposition \ref{prop:BTP-kseq} lies in the proof of \prop{2}, which can be sketched as follows.
Using the fact that there exists $(q, \alpha) \in S$ such that $|\alpha|>N_W$, we expose a partition of $S$ into two subsets $S'$ and $S''$ satisfying the \BTP{k'}, respectively the \BTP{k''}, for some $1 \leq k', k'' < k$ such that $k' + k'' \leq k$.
This is proved by using the fact that $W$ satisfies the $\BTP{k}$, and that the branching nature of the \BTP{} allows us to combine unsatisfied instances of the \BTP{} over $W_{S'}$ and $W_{S''}$ to build unsatisfied instances of the \BTP{} over $W$.
Then, since $k' < k$ and $k'' < k$, $\inter{S'}$ is $k'$-sequential and $\inter{S''}$ is $k''$-sequential by the induction hypothesis.
Finally, as $S$ is the union of $S'$ and $S''$, $W_S$ is equivalent to the union of $W_{S'}$ and $W_{S''}$, and \prop{2} follows, since $k' + k'' \leq k$.

The properties \prop{1} and  \prop{2} allow us to expose $k$ sequential weighted automata $\overline{V}_1, \ldots, \overline{V}_k$ whose union is equivalent to $W$. 
Let $U$ denote the set containing the accessible states $S$ of $D_W$ that contain only pairs $(q,\alpha)$ satisfying $|\alpha| \leq N_W$.
As there are only finitely many $\alpha \in \group$ such that $|\alpha| \leq N_W$, \prop{1} implies that $U$ is finite.
Moreover, as a consequence of $\prop{2}$, for every state $S \notin U$ in the trim part of $D_W$, $W_{S}$ can be expressed as the union of $k$ sequential weighted automata $V_i(S)$, with $1 \leq i \leq k$.
For every $1 \leq i \leq k$, let $\overline{V}_i$ be the sequential weighted automaton that copies the behaviour of $D_W$ as long as the latter stays in $U$, and swaps to $V_i(S)$ as soon as $D_W$ enters a state $S \notin U$.
Then $D_W$ is equivalent to the union of the $\overline{V}_i$, $1 \leq i \leq k$, which proves the desired result, since $D_W$ is equivalent to $W$.
Once again, the detailed proofs can be found in the appendix.


\section{Cost register automata with independent registers}
\label{sec:cra}
%

Recently, a new model of machine, named
cost register automata (CRA), has been introduced in~\cite{DBLP:conf/lics/AlurDDRY13}.
We present in this section how the class of $k$-sequential relations is also
characterized by a specific subclass of cost register automata.

A cost register automaton (CRA) \cite{DBLP:conf/lics/AlurDDRY13}
is a deterministic automaton with registers containing values from a
set $S$ and that are updated through the transitions: for
each register, its new value is computed from the old ones and from
elements of $S$ combined using some operations over $S$.
The output value is computed from the values taken by the registers at
the end of the processing of the input. Hence, a CRA defines a
relation in $A^* \times S$.

In this paper, we focus on a particular structure $(\monoid, \otimes
c)$ defined over a monoid $(\monoid, \otimes, \unit)$. In such a structure, the only
updates are unary and are of the form $X:= Y \otimes c$, where
$c \in \monoid$ and $X,Y$ are registers.
When $\monoid$ is $(\mathbb{Z},+,0)$, this class of
automata is called additive cost register
automata \cite{DBLP:conf/icalp/AlurR13}. When $\monoid$ is the free
monoid $(A^*,.,\varepsilon)$, this class is a subclass of streaming
string transducers \cite{DBLP:conf/fsttcs/AlurC10} and turns out to be
equivalent to the class of rational functions on words, \ie those
realized by finite-state transducers.

While cost register automata introduced
in~\cite{DBLP:conf/lics/AlurDDRY13} define functions from $A^*$ to
$\monoid$, we are interested in defining finite-valued relations. To
this aim,
we slightly modify the definition of CRA, allowing
to produce a set of values computed from register contents.

\begin{definition}
A \intro{cost register automaton} on the alphabet $\alphabet$ over the monoid $(\monoid,\otimes,\unit)$ is a tuple $(Q, q_{\textsl{init}}, \Reg, \delta, \mu)$ where $Q$ is a finite set of states, $q_{\textsl{init}} \in Q$ is the initial state and $\Reg$ is a finite set of registers. The transitions are given by the function
$\delta : Q \times \alphabet \to (Q \times \up(\Reg))$
where $\up(\Reg)$ is the set of functions $\Reg \to \Reg \times \monoid$ that represents the updates on the registers. Finally, $\mu$ is a finite set of $Q \times \Reg \times \monoid$ (the output relation).
\end{definition}

The semantics of such an automaton is as follows: if an update function $f$ labels a transition and $f(Y)=(X,\alpha)$, then the register $Y$ after the transition will take the value $\beta  \alpha$ where $\beta$ is the value contained in the register $X$ before the transition.  More precisely, a valuation $\nu$ is a mapping from $\Reg$ to $\monoid$ and let $\val$ be the set of such valuations.
The initial valuation $\nu_{\textsl{init}}$ is the function associating with each register the value $\unit$. A configuration is an element of $Q \times \val$. The initial configuration is $(q_{\textsl{init}},\nu_{\textsl{init}})$. A run on a word $w = w_1 \dotsm w_k \in \alphabet^*$ where for all $i$, $w_i \in \alphabet$, is a sequence of configurations $(q_1,\nu_1)(q_2,\nu_2)\ldots (q_{k+1},\nu_{k+1})$ satisfying that for all $1\leq i \leq k$, and all registers $Y$, if
$\delta(q_i,w_i)=(q_{i+1},g_i)$ with $g_i(Y)=(X,\alpha)$, then $\nu_{i+1}(Y)= \nu_{i}(X) \alpha$.  Moreover, the run is said to be accepting if $(q_1,\nu_1)$ is the initial configuration and there are $X,\alpha$ such that $(q_{k+1},X,\alpha)\in \mu$.

A cost register automaton $C$ computes a relation $\inter{C} \subseteq \alphabet^* \times \monoid$
defined by the set of pairs $(w,\nu_{k+1}(X)\alpha)$ such that $(q_1,\nu_1)(q_2,\nu_2) \ldots  (q_{k+1},\nu_{k+1})$ is an accepting run of $C$ on $w$ and
$(q_{k+1},X,\alpha)\in \mu$.

\begin{definition}
A cost register automaton is said to be with \intro{independent registers}
if for any update function $f$ which labels a transition, if $f(Y)=(X,\alpha)$ then $X=Y$.
\end{definition}

\begin{example}
\label{example:register}
Consider $\alphabet =\{a,b\}$ and $(\monoid,\otimes,\unit)=(\mathbb{Z},+,0)$. \notePA
The cost register automaton $C_0$ given in Figure~\ref{example-cra} computes the function $\flast$ introduced in Example~\ref{example:weighted}.
The register $X_a$ (\emph{resp.} $X_b$) stores the number of occurrences of the letter $a$ (\emph{resp.} $b$). Observe that these two registers are independent.
\end{example}


Independence of registers is tightly related to sequentiality of WA. We prove:

\begin{proposition}\label{prop:kseq-crak-equivalence}
  For all positive integers $k$, a relation is $k$-sequential if and only if it is computed by a cost register automaton with $k$ independent registers.
\end{proposition}

CRA are deterministic by definition, and a challenging minimisation problem
is captured by the notion of \emph{register complexity}.
It is defined for a relation as the minimal integer $k$ such that it can be defined
by a CRA with $k$ registers. By Proposition~\ref{prop:kseq-crak-equivalence},
results on the computation of the degree of sequentiality
presented in Section~\ref{sec:decision} thus also allow to compute
the register complexity for CRA with independent registers.

One can also show that the class of CRA with $k$ independent registers
is equivalent to the class of CRA with $k$ registers,
updates of the form $X:=Y\alpha$, and that are copyless
(every register appears at most once in the right-hand side of an update function).

The class of CRA with $k$ non-independent registers
was characterized in~\cite{daviaud_2016}
using the twinning property of order~$k$.
This property is weaker than our branching twinning property of order~$k$
  as it requires the same conclusion
  but only for runs labeled by the same input words.

\ifRR
\begin{example}\label{example-btp}
  We have seen that the minimal $k$ such that $W_1$ satisfies the $\BTP{k}$ is $k=4$.
  Thus it can be computed by a cost register automata with $4$ independent registers.
  One can observe that the twinning property of order $2$ from \cite{daviaud_2016} is however satisfied. Indeed, there exists
  a cost register automata with only 2 registers realizing $\flast^2$, but the two registers are not independent
  (see Appendix Section~\ref{sec:-app-exple-BTP} for details).
\end{example}
\fi

\ifRR
  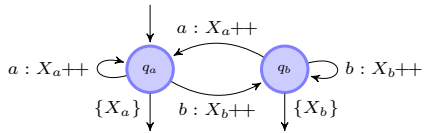
\begin{figure}[htb!]
    \centering
    \scalebox{.95}{
%
\begin{tikzpicture}[
  ->,
  >=stealth',
  shorten >=1pt,
  node distance=1.2cm,
  scale=0.8,
  every state/.style={minimum width=0.7cm,draw=blue!50,very thick,fill=blue!20, scale=0.8},
  every node/.style={font=\scriptsize},
  initial/.style={
    initial by relation,
    initial where=above,
    initial distance=.7cm,
    initial text=,
  },
  accepting/.style={
    accepting by relation,
    accepting where=below,
    accepting distance=.7cm,
    accepting text=#1,
  },
]
  \node[state, initial, accepting=$\{X_a\}$, accepting text anchor=east] (a) {$q_a$};
  \node[state, accepting=$\{X_b\}$, accepting text anchor=west] (b) [right=of a] {$q_b$};

  \path[->]     (a) edge [loop left] node         {$a: X_a\inc$}   ();
  \path[->]     (b) edge [loop right] node         {$b: X_b\inc$}   ();

  \path[bend right, ->]    (b) edge node [above] {$a: X_a\inc$}   (a);
  \path[bend right, ->]    (a) edge node [below] { $b: X_b\inc$}   (b);
\end{tikzpicture}}
    \caption{Example of a cost register automaton $C_0$ computing the function $\flast$.
    The updates are abbreviated: $X_a\inc$ means both $X_a:=X_a+1$ and $X_b:=X_b$ (and conversely).
    \label{example-cra}}
  \end{figure}
\fi


\section{The case of transducers}
\label{sec:transducers}

A transducer is defined as a weighted automaton with weights in the monoid
$B^*$. It can thus be seen as a weighted automaton with weights in the free group
$\free$.
We say that a transducer $T$ satisfies the branching twinning property
of order~$k$ if, viewed as a weighted automaton over $\free$, it satisfies the \BTP{k}.
Similarly, a relation $R \subseteq A^*  \times \alphab^*$ is said to satisfy
the Lipschitz property of order~$k$ iff it is the case when viewing $R$ as a relation
in $A^* \times \free$.

A relation $R$ of $A^* \times B^*$ is said to be \intro{positive $k$-sequential} if
it is computed by a $k$-sequential weighted automaton with weights in $B^*$ (weights on the transitions in $B^*$ and initial and final relations in $Q\times B^*$ where $Q$ is its set of states). As for the general case, it is easy to see that a relation is positive $k$-sequential if and only if it is computed by a cost register automaton with $k$ independent registers, with updates of the form $X:=Xc$ where $c \in B^*$ and with an output relation $\mu \subseteq Q\times \Reg \times B^*$.

\begin{theorem}
\label{theorem:main-transducers}
Let $T$ be a transducer from $A^*$ to $B^*$, and
$k$ be a positive integer.
The following assertions are equivalent:
\begin{enumerate}[$i)$]
\item $\inter{T}$ satisfies the Lipschitz property of order~$k$, \label{th-trans-lip}
\item $T$ satisfies the branching twinning property of order~$k$, \label{th-trans-btp}
\item $\inter{T}$ is positive $k$-sequential. \label{th-trans-positive}
\end{enumerate}
\end{theorem}

The assertions $\ref{th-trans-lip})$ and $\ref{th-trans-btp})$ are equivalent by Theorem \ref{theorem:main}. The fact that the assertion $\ref{th-trans-positive})$ implies the assertion $\ref{th-trans-btp})$ is also a consequence of Theorem \ref{theorem:main} and of the fact that
the branching twinning property of order~$k$ is a machine-independent characterization. Finally, it remains to prove that the assertion $\ref{th-trans-btp})$ implies the assertion $\ref{th-trans-positive})$.

By hypothesis, $\inter{T} \subseteq A^* \times B^*$ is computed by a transducer that satisfies the branching twinning property of order~$k$.
Thus, by Theorem \ref{theorem:main}, it is computed by a cost register automaton over $\free$ with $k$ independent registers.
We conclude using the:
\begin{proposition}
\label{proposition:positivite}
A relation in $A^* \times B^*$ is computed by a cost register automaton over $\free$ with $k$ independent registers if and only if it is computed by a cost register automaton over $B^*$ with $k$ independent registers.
\end{proposition}



\section{Decidability of \BTP{k} and computation of the sequentiality degree}
\label{sec:decision}


In this section, we prove the decidability of the following problem
  under some hypotheses on the group $\group$:\\
\noindent\textbf{The \BTP{k} Problem:}
given a weighted automaton $W$ over some group $\group$
and a number $k$, does $W$ satisfy the \BTP{k}?\\

As a corollary of Theorem~\ref{theorem:main}, this allows to
compute the degree of sequentiality for weighted automata.
We will consider two settings: first weighted automata over
some computable commutative group and second, word transducers.
%
%
%

Our decision procedures non-deterministically guess a counter-example to
the \BTP{k}. First, we show that if there exists such a counter-example
with more than $k$ loops, then there exists one with $k$ loops.
For simplicity, we can assume that the counter-example contains $k(k+1)/2$ loops \emph{i.e.}
exactly one loop per pair $(j,j')$, with $0\leq j<j' \leq k$.
%
%
This allows the procedure to first guess the "skeleton" of
the counter example, and then check that this skeleton
can be turned into a real counter-example. The skeleton
consists of the vectors of states, and,
for each pair $(j,j')$ of run indices, indicates the index
$\chi(j,j')$ of the last loop such that input words of runs $j$ and $j'$
are equal up to this loop, and the index $\eta(j,j')$ of the loop that induces
a different delay (with $\eta(j,j')\leq \chi(j,j')$).

\paragraph{Case of computable commutative groups.}
We write $W= (Q, t_{\textsl{init}}, t_{\textsl{final}}, T)$
and let $n=|Q|$.
In order to decide the branching twinning property, we will consider the $k+1$-th power of
$W$, denoted $W^{k+1}$, which accepts the set of $k+1$ synchronized runs in $W$.
We write its runs as $\vec{\rho}=(\rho_i)_{0\leq i \leq k}$ and denote by $\alpha_i$
the weight of run $\rho_i$.

\begin{theorem}\label{thm:BTP}
Let $\group=(G,\otimes)$ be a commutative group such that the operation $\otimes$
and the equality check are computable. Then the \BTP{k} problem is decidable.
\end{theorem}

\begin{proof}[Sketch]
It is easy to observe that for commutative groups, the constraint expressed on the delay in
the \BTP{k} boils down to checking that loops have different weights.

The procedure first guesses the skeleton of a counter-example as explained above.
%
The procedure then non-deterministically verifies that the skeleton can be completed into
a concrete counter-example. To this end, it uses the information stored in this skeleton
about how input words are shared between
runs (indices $\chi(j,j')$) to identify the power $p\leq k+1$ of $W$ in which the run should be identified.
The procedure is based on the two following subroutines:
\begin{itemize}
\item first, given two vectors of states $v,v'\in Q^{p}$, checking
that there exists a path from $v$ to $v'$ in $W^{p}$ is decidable,
\item second, the following problem is decidable:
given a vector of states $v\in Q^{p}$ and a pair $1\leq j\neq j' \leq p$,
check that there exists
a cycle $\vec{\rho}$ around $v$ in $W^{p}$ such that $\delay{\alpha_j}{\alpha_{j'}}\neq \unit$.
The procedure non-deterministically guesses the cycle in $W^{p}$ (its length
can be bounded by $2n^{p}$) and computes incrementally the value
of $\delay{\alpha_j}{\alpha_{j'}}$. \qed
\end{itemize}
%
\end{proof}

If we consider the group $(\mathbb{Z},+)$, we can verify
that the above procedure runs in \PSPACE if $k$ is given in unary. In addition, using
ideas similar to a lower bound proved in~\cite{DBLP:conf/icalp/AlurR13}, we can reduce
the emptiness of $k$ deterministic finite state automata to the \BTP{k} problem,
 yielding:
\begin{theorem}
Over $(\mathbb{Z},+)$, the \BTP{k} problem is \PSPACE-complete~($k$ given in unary).
\end{theorem}

\paragraph{Case of transducers.}
For word transducers, the authors of~\cite{DBLP:journals/iandc/WeberK95}
prove that a counter-example to the (classical) twinning property
is either such that loops have output words of different length,
or such that output words produced on the runs leading to the loops
have a mismatch.

Inspired by this result, we show that the skeleton described above
can be enriched with the information, for each pair of run indices $(j,j')$,
whether one should look for a loop whose output words have distinct lengths,
or for a mismatch on the paths leading to the loop.
These different properties can all be checked in \PSPACE, yielding:
\begin{theorem}
Over $(B^*,\cdot)$, the \BTP{k} problem is \PSPACE-complete ($k$ is given in unary)\footnote{The transducer is viewed as a weighted automaton over $\free$.}.
\end{theorem}

\section{Conclusion}
\label{sec:conclusion}


Multi-sequential machines are an interesting compromise between sequential
and finite-valued ones. This yields the natural problem of the minimization
of the size of the union. In this paper, we have solved this problem
for weighted automata over an infinitary finitely generated group,
a setting that encompasses standard groups.
To this end, we have introduced a new twinning property, as well as a new
Lipschitz property, and have provided an original construction from
weighted automata to $k$-sequential weighted automata,
extending the standard determinization of transducers in an intricate way.
In addition, the characterization by means of a twinning property
allows to derive efficient decision procedures, and all our results are also valid for
word transducers.

As a complement, these results can be generalized to non finitely generated
groups, using ideas similar to those developed in~\cite{daviaud_2016}.
As future work, we plan to lift these results to other settings, like
infinite or nested words. Another challenging research direction consists in considering
other operations to aggregate weights of runs.










\ifRR

\newpage

\appendix

\section*{Appendix}

In all of the Appendix, $A$ denotes a finite alphabet, $\group$ denotes an infinitary finitely generated group and $\gene$ a finite set of generators of $\group$.


\section{Proofs of Section~\ref{sec:tp}: Branching twinning and Lipschitz properties}
%

\subsection{Example}\label{sec:-app-exple-BTP}

\begin{example}
The weighted automaton $W_1$ obtained by concatenating $W_0$ with itself is depicted in Figure~\ref{example-cra-flast-squared} (left).
This WA realizes the function $\flast^2$.
  As outlined in Example~\ref{example-btp}, $\flast^2$ can be realized by a CRA with two registers.
  Indeed, the weighted automaton $W_1$ can be shown to satisfy the $\TP{2}$ of \cite{daviaud_2016}.
  Figure~\ref{example-cra-flast-squared} (right) shows such a $2$-register machine $C_1$.
  (See \cite{daviaud_2016} for more details.)

  Note, however, that the two registers used in $C_1$ are not independent.
  (See the transitions reading $\#$.)
  Actually, we need at least $4$ independent registers to handle words
    like $a^na\#a^ma$, $a^nb\#a^ma$, $a^nb\#b^ma$ and $a^na\#b^ma$
    which can be used to produce values arbitrarily far one to another.
  Finally, $4$ independent registers are enough to realize $\flast^2$ as we only have to guess
    one of the four combinations of last letters of the two words.
\end{example}

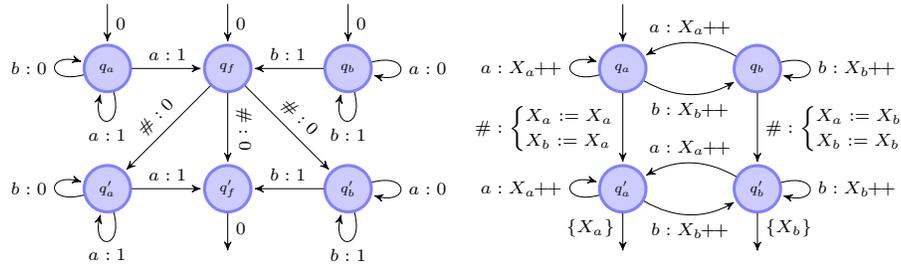
\begin{figure}[htb!]
  \begin{minipage}[t]{.5\textwidth}
    \centering
    \vspace{0pt}\scalebox{.95}{
%
\def\dh{1cm}
\def\dv{1.cm}
\begin{tikzpicture}[
  ->,
  >=stealth',
  shorten >=1pt,
  node distance=\dh,
  scale=0.8,
  every state/.style={minimum width=0.7cm,draw=blue!50,very thick,fill=blue!20, scale=0.8},
  every node/.style={font=\scriptsize},
  initial/.style={
    initial by relation,
    initial where=above,
    initial distance=.7cm,
    initial text=#1,
    initial text anchor=base west,
  },
  accepting/.style={
    accepting by relation,
    accepting where=below,
    accepting distance=.7cm,
    accepting text=#1,
    accepting text anchor=west,
  },
]
  \node[state, initial=$0$] (a) [] {$q_a$};
  \node[state, initial=$0$] (f) [right=of a] {$q_f$};
  \node[state, initial=$0$] (b) [right=of f] {$q_b$};

  \node[state] (a2) [below=\dv of a] {$q'_a$};
  \node[state, accepting=$0$] (f2) [below=\dv of f] {$q'_f$};
  \node[state] (b2) [below=\dv of b] {$q'_b$};

  \path[->]     (a) edge [loop left] node         {$b: 0$}   ();
  \path[->]     (a) edge [loop below] node         {$a: 1$}   ();
  \path[->]     (a) edge []  node [above, sloped] {$a: 1$}   (f);

  \path[->]     (a2) edge [loop left] node         {$b: 0$}   ();
  \path[->]     (a2) edge [loop below] node         {$a: 1$}   ();
  \path[->]     (a2) edge []  node [above, sloped] {$a: 1$}   (f2);

  \path[->]     (f) edge []  node [above, sloped] {$\#: 0$}   (a2);
  \path[->]     (f) edge []  node [above, sloped] {$\#: 0$}   (f2);
  \path[->]     (f) edge []  node [above, sloped] {$\#: 0$}   (b2);

  \path[->]     (b) edge [loop right] node         {$a: 0$}   ();
  \path[->]     (b) edge [loop below] node         {$b: 1$}   ();
  \path[->]     (b) edge []  node [above, sloped] {$b: 1$}   (f);

  \path[->]     (b2) edge [loop right] node         {$a: 0$}   ();
  \path[->]     (b2) edge [loop below] node         {$b: 1$}   ();
  \path[->]     (b2) edge []  node [above, sloped] {$b: 1$}   (f2);

\end{tikzpicture}}
  \end{minipage}
  \begin{minipage}[t]{.5\textwidth}
    \centering
    \vspace{0pt}\scalebox{.95}{
%
\def\dh{1.2cm}
\def\dv{1.cm}
\begin{tikzpicture}[
  ->,
  >=stealth',
  shorten >=1pt,
  node distance=\dh,
  scale=0.8,
  every state/.style={minimum width=0.7cm,draw=blue!50,very thick,fill=blue!20, scale=0.8},
  every node/.style={font=\scriptsize},
  initial/.style={
    initial by relation,
    initial where=above,
    initial distance=.7cm,
    initial text=,
  },
  accepting/.style={
    accepting by relation,
    accepting where=below,
    accepting distance=.7cm,
    accepting text=#1,
  },
]
  \node[state, initial] (a) {$q_a$};
  \node[state] (b) [right=of a] {$q_b$};

  \node[state, accepting=$\{X_a\}$, accepting text anchor=east] (a2) [below=\dv of a] {$q'_a$};
  \node[state, accepting=$\{X_b\}$, accepting text anchor=west] (b2) [right=of a2] {$q'_b$};

  \path[->]     (a) edge [loop left] node         {$a: X_a\inc$}   ();
  \path[->]     (b) edge [loop right] node         {$b: X_b\inc$}   ();

  \path[bend right, ->]    (b) edge node [above] {$a: X_a\inc$}   (a);
  \path[bend right, ->]    (a) edge node [below] { $b: X_b\inc$}   (b);

  \path[->]    (a) edge node [left,xshift=3mm] {$\#: \begin{cases} X_a := X_a\\ X_b := X_a \end{cases}$}   (a2);
  \path[->]    (b) edge node [right] { $\#: \begin{cases} X_a := X_b\\ X_b := X_b \end{cases}$}   (b2);

  \path[->]     (a2) edge [loop left] node         {$a: X_a\inc$}   ();
  \path[->]     (b2) edge [loop right] node         {$b: X_b\inc$}   ();

  \path[bend right, ->]    (b2) edge node [above] {$a: X_a\inc$}   (a2);
  \path[bend right, ->]    (a2) edge node [below] { $b: X_b\inc$}   (b2);
\end{tikzpicture}}
  \end{minipage}
  \caption{A weighted automaton $W_1$ (left)
    and a cost register automaton $C_1$ (right) with 2 registers that compute $\flast^2$.
  The updates are abbreviated: $X_a\inc$ means both $X_a:=X_a+1$ and $X_b:=X_b$ (and conversely).}
  \label{example-cra-flast-squared}
\end{figure}

\subsection{Lipschitz property of order $k$}
We prove that $k$-sequential WA satisfies the Lipschitz property of order~$k$. It is given  by Proposition~\ref{prop:kseq-lip}.


\subsection*{Proof of Proposition~\ref{prop:kseq-lip}}

Consider a weighted automaton $\poids$ defined as the union of $k$ sequential weighted automata $\poids_1, \ldots, \poids_k$. Let $(w_0,\alpha_0), \ldots, (w_k,\alpha_k) \in \inter{\poids}$. By the pigeon hole principle, there are $1\leq j< j'\leq k$ and $1\leq i \leq k$ such that $(w_j,\alpha_j) \in \inter{\poids_i}$ and $(w_{j'},\alpha_{j'}) \in \inter{\poids_i}$. The result follows by sequentiality of $\poids_i$: there is a unique computation on the longest common prefix of $w_j$ and $w_{j'}$ in $\poids_i$, thus:
\[\cayley{\alpha_j}{\alpha_{j'}} \leq 2 \base + \base \dist(w_j,w_{j'}) + 2 \base \leq 4\base(\dist(w_j,w_{j'})+1)\]

\subsection{An alternative branching twinning property}
\label{appendix_loops_BTP}


Let us consider a similar definition of the $\BTP{k}$ by increasing the number of loops.
This leads to an alternative branching twinning property, that we will call $\BTP{k}'$, obtained from the $\BTP{k}$ by 
requiring the property not only for $k$ cycles, but for $m$ cycles, for every $m \geq k$. We prove in Lemma \ref{lemma-BTP-loops} that $\BTP{k}$ and 
$\BTP{k}'$ are equivalent.

\begin{definition}
  \label{def:tp'} A weighted automaton over a group $\group$ satisfies $\BTP{k}'$ if:
  \begin{itemize}
    \item for all $m \geq k$
    \item for all states $\{q_{i,j} \mid \ini{i}{0}{m} \text{ and } \ini{j}{0}{k}\}$
      with $q_{0,j}$ initial for all $j$,
    \item for all $\gamma_j \in \group$ such that $(q_{0,j}, \gamma_j) \in t_{\textsl{init}}$ with $\ini{j}{0}{k}$,
    \item for all words $u_{i,j}$ and $v_{i,j}$ with $\ini{i}{1}{m} \text{ and } \ini{j}{0}{k}$ such that
      there are $k+1$ runs satisfying for all $\ini{i}{1}{m}$, for all $\ini{j}{0}{k}$,
        $q_{i-1,j} \xrightarrow{u_{i,j} \mid \alpha_{i,j}} q_{i,j}$ and
        $q_{i,j} \xrightarrow{v_{i,j} \mid \beta_{i,j}} q_{i,j}$ (see Figure~\ref{figure:BTPk'}),
  \end{itemize}
  there are $j \neq j'$ such that for all $\ini{i}{1}{m}$,
  if for every $1\leq i' \leq i$, we have $u_{i',j}=u_{i',j'}$ and $v_{i',j}=v_{i',j'}$,
  then we have
  \begin{align*}
    \delay{\gamma_j \alpha_{1,j} \dotsm \alpha_{i,j}}
      {\gamma_{j'} \alpha_{1,j'} \dotsm \alpha_{i,j'}} =
    \delay{\gamma_j \alpha_{1,j} \dotsm \alpha_{i,j} \beta_{i,j}}
      {\gamma_{j'} \alpha_{1,j'}  \dotsm \alpha_{i,j'} \beta_{i,j'}}.
  \end{align*}
\end{definition}

\begin{figure}[hb!]
  \centering
  \def \loopCount {m}
\scalebox{.8}{
\begin{tikzpicture}[
  ->,
  >=stealth',
  shorten >=1pt,
  auto,
  node distance=2.8cm,
  every state/.style={minimum width=1cm,draw=blue!50,very thick,fill=blue!20},
  every node/.style={font=\small,scale=0.8},
]
  \tikzstyle{every edge}=[draw=black,font=\small]

  \foreach \j/\y in {0/0,1/-1.8,k/-4.1}
  {
    \node[initial=$\weight{\gamma_\j}$,state] (A0) at (0,\y) {$q_{0,\j}$};

    \foreach \i/\x in {1/2.5,2/5,\loopCount/9} {
      \node[state] (A\i)  at (\x,\y) {$q_{\i,\j}$};

      \path (A\i) edge [loop above] node {\trans{v_{\i,\j}}{\beta_{\i,\j}}} (A\i);
    }

    \foreach \im/\i in {0/1,1/2} {
      \path (A\im) edge node {\trans{u_{\i,\j}}{\alpha_{\i,\j}}} (A\i);
    }
    \foreach \im/\i in {2/\loopCount} {
      \path (A\im) edge [dashed] node {} (A\i);
    }
  }

  \node[rotate=90] (textleft) at (-2,-2) {$k+1$ runs};
\end{tikzpicture}
}
  \caption{Equivalent definition of the branching twinning property of order~$k$\label{figure:BTPk'}}
\end{figure}
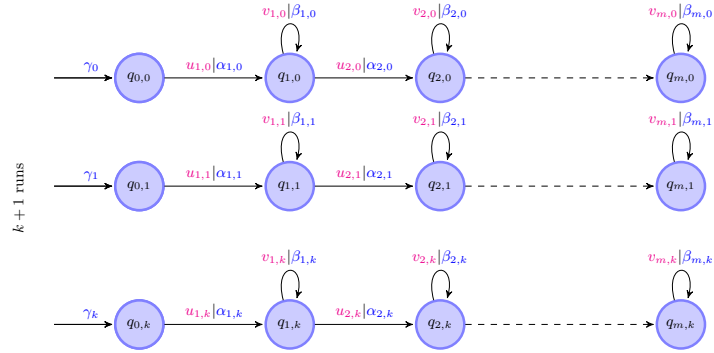

\begin{lemma}\label{lemma-BTP-loops}
  For all positive integer k, a weighted automaton satisfies $\BTP{k}$ if and only if it satisfies $\BTP{k}'$
\end{lemma}

\begin{proof}
  By definition, $\BTP{k}'$ implies $\BTP{k}$. For the converse implication, suppose $\BTP{k}'$ is not satisfied. That is, there are:
  \begin{itemize}
    \item an integer $m \geq k$
    \item states $\{q_{i,j} \mid \ini{i}{0}{m} \text{ and } \ini{j}{0}{k}\}$
      with $q_{0,j}$ initial for all $j$,
    \item elements $\gamma_j \in \group$ such that $(q_{0,j}, \gamma_j) \in t_{\textsl{init}}$ with $\ini{j}{0}{k}$,
    \item words $u_{i,j}$ and $v_{i,j}$ with $\ini{i}{1}{m} \text{ and } \ini{j}{0}{k}$ such that
      there are $k+1$ runs satisfying for all $\ini{i}{1}{m}$, for all $\ini{j}{0}{k}$,
        $q_{i-1,j} \xrightarrow{u_{i,j} \mid \alpha_{i,j}} q_{i,j}$ and
        $q_{i,j} \xrightarrow{v_{i,j} \mid \beta_{i,j}} q_{i,j}$ (see Figure~\ref{figure:BTPk'}),
  \end{itemize}
  such that for all $j \neq j'$ there exists $\ini{i}{1}{m}$, also denoted $\zeta_{j,j'}$, such that
  \begin{itemize}
    \item for every $1\leq i' \leq i$, we have $u_{i',j}=u_{i',j'}$ and $v_{i',j}=v_{i',j'}$, but
    \item
      $\delay{\gamma_j \alpha_{1,j} \dotsm \alpha_{i,j}}
        {\gamma_{j'} \alpha_{1,j'} \dotsm \alpha_{i,j'}} \neq
      \delay{\gamma_j \alpha_{1,j} \dotsm \alpha_{i,j} \beta_{i,j}}
        {\gamma_{j'} \alpha_{1,j'}  \dotsm \alpha_{i,j'} \beta_{i,j'}}$.
  \end{itemize}

  We prove now that we can only consider $k$ loops and still preserve the property. We inductively build a partition $P_i$, for all $\ini{i}{0}{m}$, of the set $\{0,\ldots,k\}$ of run indexes:
  \begin{itemize}
    \item $P_0 = \{\{0,\ldots,k\}\}$,
    \item $P_{i+1}$ refines $P_i$ such that $j$ and $j'$ remains in the same class if and only if
    \begin{align*}
      & \delay{\gamma_j \alpha_{1,j} \dotsm \alpha_{i+1,j}}
        {\gamma_{j'} \alpha_{1,j'} \dotsm \alpha_{i+1,j'}} \\
      =\; & \delay{\gamma_j \alpha_{1,j} \dotsm \alpha_{i+1,j} \beta_{i+1,j}}
        {\gamma_{j'} \alpha_{1,j'}  \dotsm \alpha_{i+1,j'} \beta_{i+1,j'}}
    \end{align*}
  \end{itemize}

  We know that $P_m$ is the set of singleton sets. Moreover, since the partioned set contains $k+1$ elements, there are at most $k$ indexes $\ini{i}{1}{m}$ such that $P_{i-1} \neq P_{i}$. For all $j \neq j'$, consider $i_s$ the smallest index such that $j$ and $j'$ are not in the same class in $P_{i_s}$.
  In particular, $i_s$ is the smallest $i$ such that
  \begin{align*}
    \delay{\gamma_j \alpha_{1,j} \dotsm \alpha_{i,j}}
      {\gamma_{j'} \alpha_{1,j'} \dotsm \alpha_{i,j'}} \neq
    \delay{\gamma_j \alpha_{1,j} \dotsm \alpha_{i,j} \beta_{i,j}}
      {\gamma_{j'} \alpha_{1,j'}  \dotsm \alpha_{i,j'} \beta_{i,j'}}.
  \end{align*}

  But then $i_s \leq \zeta_{j,j'}$
  and thus for every $1\leq i' \leq i_s$, we have $u_{i',j}=u_{i',j'}$ and $v_{i',j}=v_{i',j'}$.
  This proves that the $\BTP{k}$ is not satisfied either.
\qed\end{proof}

\subsection{Equivalence of Lipschitz and branching twinning properties: proof of Proposition \ref{prop:BTP-lip-equivalence}}
\label{appendix_BTP_lip}
\paragraph*{\BTP{k} implies \lip{k}.}


Let $W$ denote a weighted automaton over $\group$ and $Q$ its set of states.


We define an object, called a \chunk{k,m}, that has a shape that resembles to an input of the $\BTP{k}'$ (with no initial outputs and additional runs at the right end after the last loops). Chunks will be assembled to build larger chunks.

\begin{definition}[\chunk{k,m}]
  A \emph{\chunk{k,m}} is a split of $k$ runs $(\rho_j)_{\ini{j}{1}{k}}$ in $2m+1$ parts, $m$ of which are loops, that is such that there are
  \begin{itemize}
    \item states $\{q_{i,j} \mid \ini{i}{0}{m+1}, \text{ and } \ini{j}{1}{k}\}$
    \item words $u_{i,j}$ for all $\ini{i}{1}{m+1}$ and $\ini{j}{1}{k}$,
    \item words $v_{i,j}$ for all $\ini{i}{1}{m}$ and $\ini{j}{1}{k}$
  \end{itemize}
  such that, for all $\ini{j}{1}{k}$, we have
  \begin{itemize}
    \item $\rho_{i,j} = q_{i-1,j} \xrightarrow{u_{i,j} \mid \alpha_{i,j}} q_{i,j}$ and
      $\rho'_{i,j} = q_{i,j} \xrightarrow{v_{i,j} \mid \beta_{i,j}} q_{i,j}$, for all $\ini{i}{1}{m+1}$,
    \item $\rho_{j} = \rho_{1,j} \rho'_{1,j} \cdots \rho_{m,j} \rho'_{m,j} \rho_{m+1,j}$,
  \end{itemize}
  The $\rho_{i,j}$ are the backbone of the chunk and the $\rho'_{i,j}$ are its loops. The \emph{backbone length} of this chunk is $\max_{\ini{j}{1}{k}} \sum_{\ini{i}{1}{m+1}} u_{i,j}$.
\end{definition}


The read words and produced weights inside a chunk may be the empty word and the neutral element of $\group$. Thus, if some of the runs are shorter in a chunk, they can be completed at will with trivial $\varepsilon$ loops (that produce $\unit$).

Also, one can join two chunks $\mathcal{C}_1$ and $\mathcal{C}_2$ when the last states of $\mathcal{C}_1$ are the same as the first states of $\mathcal{C}_2$. By fusing the last part of the runs of $\mathcal{C}_1$ with the first part of the runs of $\mathcal{C}_2$, we obtain a new chunk whose number of loops is the sum of the number of loops of $\mathcal{C}_1$ and $\mathcal{C}_2$, and whose backbone length is the sum of the backbone lengths of $\mathcal{C}_1$ and $\mathcal{C}_2$.


\begin{lemma}\label{lem-split-run}
  Let $\rho$ be a run in $W$ and let $n = |Q|$. If $|in(\rho)| \geq |Q|$ then there exist $\rho_0,\ldots,\rho_n$ and $\rho'_1,\ldots,\rho'_n$, such that $\rho = \rho_0 \rho'_1 \rho_1 \ldots \rho'_n \rho_n$, all the $\rho'_i$ are loops and $|in(\rho_0\ldots\rho_n)| < |Q|$.
\end{lemma}

\begin{lemma}\label{lem-run-to-chunk}
  From $k$ runs in $W$ ($k \geq 1$), we can build a $\chunk{k,k|Q|^k}$ whose backbone length is $k|Q|^k$-bounded.
\end{lemma}

\begin{proof}
  We proceed by induction on the number $k$ of runs $\rho_j$ ($\ini{j}{1}{k}$).

  For the base case ($k=1$), either $|in(\rho_1)| < |Q|$ and we can build a \chunk{1,0} $B = \left( \rho_1 \right)$ which can be completed to a \chunk{1,|Q|}, or $|in(\rho_1)| \geq |Q|$ and by Lemma~\ref{lem-split-run} we can find a split of $\rho_1$ that gives a \chunk{1,|Q|}.
  In both cases, the backbone length of this \chunk{1,|Q|} is $|Q|$-bounded.

  For the induction case ($k>1$), let $w = \lcp_{\ini{j}{1}{k}} \{ in(\rho_j) \}$ and let $\rho'_j,\rho''_j$ for all $\ini{j}{1}{k}$ such that $\rho_j = \rho'_j\rho''_j$ and $in(\rho'_j) = w$.
  We will first build a \chunk{k,|Q|^k} for the $k$ runs $\rho'_j$ ($\ini{j}{1}{k}$) and then join on its right end some chunks for the $\rho''_j$ runs ($\ini{j}{1}{k}$).

  Either $|w| < |Q|^k$ and we can build a \chunk{k,0} $\left( \rho'_j \right)_{\ini{j}{1}{k}}$ that can be completed to a \chunk{k,|Q|^k} $B'$.
  Or $|w| \geq |Q|^k$ and, by Lemma~\ref{lem-split-run} applied on $W^k$, for all $\ini{j}{1}{k}$ there exist $\rho'_{0,j},\ldots,\rho'_{2n,j}$, where $n = |Q|^k$, such that
    $\rho'_j = \rho'_{0,j} \ldots \rho'_{2n,j}$,
    for all $\ini{i}{0}{n-1}$ the $\rho'_{2i+1,j}$ are loops and
    $|in(\rho'_0\rho'_2\ldots\rho'_{2n})| < |Q|^k$.
  This gives us a \chunk{k,|Q|^k} $B' = \left( \rho'_{i,j} \right)_{\ini{i}{0}{2n},\ini{j}{1}{k}}$. In both cases, the backbone of this \chunk{k,|Q|^k} is $|Q|^k$-bounded.

  We partition the runs $(\rho''_j)_{\ini{j}{1}{k}}$ by the first letter they read. Each member of the partition is of size strictly less than $k$ and their contained runs read inputs with a common prefix. We can apply the induction hypothesis to exhibit, for each member of the partition, a chunk that can be completed to a \chunk{k-1,(k-1)|Q|^{k-1}} and has a backbone length bounded by $(k-1)|Q|^{k-1}$. By adequately (w.r.t. the order of the partition) joining those chunks on the right end of $B'$ we obtain a \chunk{k,|Q|^k + (k-1)|Q|^{k-1}} that has a backbone length bounded by $|Q|^k + (k-1)|Q|^{k-1} \leq k|Q|^k$. Furthermore, it can be completed to a \chunk{k,k|Q|^k}.
\qed\end{proof}

\begin{lemma}
If $W$ satisfies $\BTP{k}$ then $\inter{W}$ satisfies $\lip{k}$.
\end{lemma}

\begin{proof}
  If $W$ satisfies $\BTP{k}$ then it also satisfies $\BTP{k}'$.

  Let $(w_0,\alpha_0),\ldots,(w_k,\alpha_k) \in \inter{W}$.
  There are initial states $p_0,\ldots,p_k \in Q$,
    final states $q_0,\ldots,q_k \in Q$,
    runs $\rho_0,\ldots,\rho_k$,
    elements $\eta_j,\zeta_j,\theta_j \in \group$ for $\ini{j}{0}{k}$ such that,
    for all $\ini{j}{0}{k}$, $\alpha_j = \eta_j \zeta_j \theta_j$,
      $\rho_j = p_j \xrightarrow{w_j|\zeta_j} q_j$,
      $(p_j,\eta_j) \in \tinit$ and $(q_j,\theta_j) \in \tfinal$.

  Let $m = (k+1)|Q|^{k+1}$.
  By Lemma~\ref{lem-run-to-chunk}, we can build a $\chunk{k+1, m}$ whose backbone length is bounded by $m$. That is there are:
  \begin{itemize}
    \item words $u_{i,j}$ for all $\ini{i}{0}{m}, \ini{j}{0}{k}$ and words $v_{i,j}$ for $\ini{i}{1}{m}, \ini{j}{0}{k}$
      such that $w_j = u_{0,j} v_{1,j} u_{1,j} \cdots v_{m,j} u_{m,j}$ and $|u_{0,j} u_{1,j} \cdots u_{m,j}| < m$
    \item states $q_{i,j}$ for all $\ini{i}{0}{m+1}, \ini{j}{0}{k}$
      such that $q_{0,j} = p_j$ and $q_{m+1,j} = q_j$,
    \item elements of $\group$ $\alpha_{i,j}$ for all $\ini{i}{0}{m}, \ini{j}{0}{k}$
      and $\beta_{i,j}$ for all $\ini{i}{0}{m}, \ini{j}{1}{k}$
      such that $\zeta_j = \alpha_{0,j} \beta_{1,j} \alpha_{1,j} \cdots \beta_{m,j} \alpha_{m,j}$,
  \end{itemize}
  such that there are $k+1$ runs satisfying:
  \begin{itemize}
    \item for all $\ini{i}{0}{m}, \ini{j}{0}{k}$, $q_{i,j} \xrightarrow{u_{i,j} \mid \alpha_{i,j}} q_{i+1,j}$ and
    \item for all $\ini{i}{1}{m}, \ini{j}{0}{k}$, $q_{i,j} \xrightarrow{v_{i,j} \mid \beta_{i,j}} q_{i,j}$.
  \end{itemize}

  As $W$ satisfies the $\BTP{k}'$, there are $j \neq j'$ such that for all $\ini{i}{1}{m}$,
  if for every $1\leq i' \leq i$, we have $u_{i'-1,j}=u_{i'-1,j'}$ and $v_{i',j}=v_{i',j'}$, then we have
  \begin{align*}
    &\delay{\eta_{j} \alpha_{0,j} \dotsm \alpha_{i-1,j}}
        {\eta_{j'} \alpha_{0,j'} \dotsm \alpha_{i-1,j'}} \\
    =\; &\delay{\eta_{j} \alpha_{0,j} \dotsm \alpha_{i-1,j} \beta_{i,j}}
        {\eta_{j'} \alpha_{0,j'}  \dotsm \alpha_{i-1,j'} \beta_{i,j'}}.
  \end{align*}

  Let $\ini{i}{1}{m}$ be the minimum index such that $u_{i-1,j} \neq u_{i-1,j'}$ or $v_{i,j} \neq v_{i,j'}$. Then
  \begin{align*}
    & \delay{\eta_{j} \zeta_{j}}{\eta_{j'} \zeta_{j'}} \\
    =\; & \delay{\eta_{j} \alpha_{0,j} \beta_{1,j} \alpha_{1,j} \cdots \beta_{m,j} \alpha_{m,j}}
        {\eta_{j'} \alpha_{0,j'} \beta_{1,j'} \alpha_{1,j'} \cdots \beta_{m,j'} \alpha_{m,j'}} \\
    =\; & \delay{\eta_{j} \alpha_{0,j} \alpha_{1,j} \cdots \beta_{m,j} \alpha_{m,j}}
        {\eta_{j'} \alpha_{0,j'} \alpha_{1,j'} \cdots \beta_{m,j'} \alpha_{m,j'}} \\
    & \vdots \\
    =\; & \delay{\eta_{j} \alpha_{0,j} \dotsm \alpha_{i-1,j} \beta_{i,j} \alpha_{i,j} \cdots \beta_{m,j} \alpha_{m,j}}
        {\eta_{j'} \alpha_{0,j'} \dotsm \alpha_{i-1,j'} \beta_{i,j'} \alpha_{i,j'} \cdots \beta_{m,j'} \alpha_{m,j'}}.
  \end{align*}

  Also we have that
  \begin{align*}
    \dist(w_{j},w_{j'})
    & = \dist(u_{i-1,j} v_{i,j} u_{i,j} \cdots v_{m,j} u_{m,j},
                u_{i-1,j'} v_{i,j'} u_{i,j'} \cdots v_{m,j'} u_{m,j'})
  \end{align*}

  Therefore
  \begin{align*}
    & \cayley{\unit}{\alpha_{i-1,j} \beta_{i,j} \alpha_{i,j} \cdots \beta_{m,j} \alpha_{m,j}}
      + \cayley{\unit}{\alpha_{i-1,j'} \beta_{i,j'} \alpha_{i,j'} \cdots \beta_{m,j'} \alpha_{m,j'}} \\
    & \leq
      \base \dist(u_{i-1,j} v_{i,j} u_{i,j} \cdots v_{m,j} u_{m,j},
                          u_{i-1,j'} v_{i,j'} u_{i,j'} \cdots v_{m,j'} u_{m,j'}) \\
    & \leq \base \dist(w_{j},w_{j'})
  \end{align*}

  Since $|u_{0,j} u_{1,j} \cdots u_{i-2,j}| < m$ and $|u_{0,j'} u_{1,j'} \cdots u_{i-2,j'}| < m$, we have that
  \begin{align*}
    \cayley{\alpha_{j}}{\alpha_{j'}}
    & \leq 2 \base m + \base \dist(w_{j},w_{j'}) + 2 \base \\
    & \leq 2 \base (m + 1) (\dist(w_{j},w_{j'}) + 1)
  \end{align*}

  This satisfies the lipschitz property of order~$k$ for $\Kc = 2 \base ((k+1) |Q|^{k+1} + 1)$.
\qed\end{proof}

\paragraph*{\lip{k} implies \BTP{k}.}

Consider a weighted automaton $\poids$ that does not satisfy \BTP{k}. Let us prove that $\inter{W}$ does not satisfy \lip{k}. It is a consequence of the following Lemma.

\begin{lemma}
\label{lemma:nonBTP_nonlip}
If $\poids$ does not satisfy \BTP{k}, then for all positive integers $\Kc$, there are $k+1$ words $w_0,\ldots, w_k$, initial states $q_0,\ldots, q_k$, with $(q_0,\gamma_0),\ldots ,(q_k,\gamma_k) \in t_{\textsl{init}}$, states $p_0,\ldots, p_k$ and $k+1$ runs:
\[q_j \xrightarrow{w_j\mid \alpha_j} p_j \quad \text{for all }j\in \{0, \ldots, k\},\]
such that for all $j \neq j'$, $\cayley{\gamma_j \alpha_j}{\gamma_{j'} \alpha_{j'}} > \Kc (\dist(w_j,w_{j'})+1)$.
\end{lemma}

\begin{proof}
The idea behind the proof is to consider a witness as described in Figure~\ref{figure:BTPk}. If \BTP{k} is not satisfied, then one can pump the loops "the right number of times" to: (1) sufficiently increase the caley distance between the weights of the runs, (2) not increase to much the distance between the corresponding labelling words.

Let $\Kc$ be a positive integer. Since $\poids$ does not satisfy \BTP{k}, then there are:
\begin{itemize}
\item states $\{q_{i,j} \mid i,j \in \{0,\dots, k\}\}$ with $q_{0,j}$ initial for all $j$,
\item pairs $(q_{0,j}, \gamma_j) \in t_{\textsl{init}}$ for $j \in \{0,\dots, k\}$,
\item words $u_{i,j}$ and $v_{i,j}$ with $i,j\in\{1,\ldots,k\}$ and $k+1$ runs
\[q_{i-1,j} \xrightarrow{u_{i,j} \mid \alpha_{i,j}} q_{i,j} \quad \text{and} \quad q_{i,j} \xrightarrow{v_{i,j} \mid \beta_{i,j}} q_{i,j} \quad \text{for} \quad 0 \leq j \leq k,\ 1 \leq i \leq k \]
\end{itemize}
such that for all $j \neq j'$, there is $i \in \{1,\ldots, k\}$ such that for all $1\leq i' \leq i$, we have $u_{i',j}=u_{i',j'}$, $v_{i',j}=v_{i',j'}$ and:
\[\delay{\gamma_j \alpha_{1,j} \dotsm \alpha_{i,j}}{\gamma_{j'} \alpha_{1,j'} \dotsm \alpha_{i,j'}}
\neq \delay{\gamma_j \alpha_{1,j} \dotsm \alpha_{i,j} \beta_{i,j}}{\gamma_{j'} \alpha_{1,j'}  \dotsm \alpha_{i,j'} \beta_{i,j'}}\]

We construct by induction (in decreasing order) a sequence of positive integers $t_k, \ldots, t_1$. Let us give the construction of $t_i$.
Let $L_i$ be the maximal length of the words $u_{i+1,j}v_{i+1,j}^{t_{i+1}} \dotsm u_{k,j} v_{k,j}^{t_k}$ over all $0 \leq j \leq k$.
Consider $T_i$ the set of pairs $(j,j')$ such that for all $1\leq i' \leq i$, we have $u_{i',j}=u_{i',j'}$, $v_{i',j}=v_{i',j'}$ and:
\[\delay{\gamma_j \alpha_{1,j} \dotsm \alpha_{i,j}}{\gamma_{j'} \alpha_{1,j'} \dotsm \alpha_{i,j'}}
\neq \delay{\gamma_j \alpha_{1,j} \dotsm \alpha_{i,j} \beta_{i,j}}{\gamma_{j'} \alpha_{1,j'}  \dotsm \alpha_{i,j'} \beta_{i,j'}}\]

One can choose an integer $N$ such that for all pairs $(j,j') \in T_i$,
\[\cayley{\gamma_{j} \alpha_{1,j} \dotsm \alpha_{i,j} (\beta_{i,j})^N}{\gamma_{j'} \alpha_{1,j'} \dotsm \alpha_{i,j'} (\beta_{i,j'})^N} > 2L_i(\base + \Kc) + \Kc\]
Set $t_i=N$.

The words $w_j = u_{1,j}v_{1,j}^{t_{1}} \dotsm u_{k,j} v_{k,j}^{t_k}$ and the corresponding runs fulfil the condition of the Lemma. Indeed, let $j\neq j'$, and $i$ the minimal index such that $(j,j') \in T_i$.
Such an index $i$ exists by hypothesis.
For $\ell \in \{j,j'\}$, set $\alpha_\ell = \alpha_{1,\ell} (\beta_{1,\ell})^{t_1} \alpha_{2,\ell} \dotsm  \alpha_{k,\ell} (\beta_{k,\ell})^{t_k}$ and $\overline{\alpha}_\ell =  \alpha_{1,\ell} \alpha_{2,\ell} \dotsm \alpha_{{i-1},\ell} \alpha_{i,\ell} (\beta_{i,\ell})^{t_i}$.

We have:
\begin{align*}
& \delay{\gamma_{j} \alpha_j}{\gamma_{j'} \alpha_{j'}} \\
= \ & \delay{\gamma_{j} \overline{\alpha}_j \alpha_{{i+1},j} \dotsm \alpha_{k,j} (\beta_{k,j})^{t_k}}{\gamma_{j'} \overline{\alpha}_{j'} \alpha_{{i+1},j'} \dotsm \alpha_{k,j'} (\beta_{k,j'})^{t_k}} \\
= \ & (\alpha_{{i+1},j} (\beta_{i+1,j})^{t_{i+1}} \dotsm \alpha_{k,j} (\beta_{k,j})^{t_k})^{-1} \delay{\gamma_{j}\overline{\alpha}_j}{\gamma_{j'}\overline{\alpha}_{j'}} \alpha_{i+1,j'} (\beta_{i+1,j'})^{t_{i+1}} \dotsm \alpha_{k,j'} (\beta_{k,j'})^{t_k}
\end{align*}

Moreover, $\dist(w_j,w_{j'}) \leq 2L_i$ by definition of $L_i$.
Then:
\[\cayley{\gamma_{j}\alpha_j}{\gamma_{j'}\alpha_{j'}} > 2L_i(\base + \Kc) + \Kc - 2 \base L_i \geq \Kc(\dist(w_j,w_{j'})+1)\]
\qed\end{proof}

\begin{lemma}
If $\poids$ does not satisfy \BTP{k} then $\inter{W}$ does not satisfy \lip{k}.
\end{lemma}

\begin{proof}
Let $\Kc$ be a positive interger and $\Kc' = \Kc (2N+1)$ where $N$ is the number of states of $\poids$.
By Lemma \ref{lemma:nonBTP_nonlip}, there are $k+1$ words $w_0,\ldots, w_k$, initial states $q_0,\ldots, q_k$, with $(q_0,\gamma_0),\ldots ,(q_k,\gamma_k) \in t_{\textsl{init}}$, states $p_0,\ldots, p_k$ and $k+1$ runs:
\[q_j \xrightarrow{w_j\mid \alpha_j} p_j \quad \text{for all }j\in \{0, \ldots, k\},\]
such that for all $j \neq j'$, $\cayley{\gamma_{j} \alpha_j}{\gamma_{j'} \alpha_{j'}} > \Kc'(\dist(w_j,w_{j'})+1)$.
These $k+1$ runs can be completed into accepting runs ending in accepting states $r_0,\ldots, r_k$ labeled by $w_0w_0',\ldots, w_kw_k'$ with weights $\alpha_0\alpha_0',\ldots, \alpha_k\alpha_k'$ such that for all $0\leq j \leq k$,
$|w_j'| \leq N$. Consider $(r_0,\beta_0),\ldots ,(r_k,\beta_k) \in t_{\textsl{final}}$.

For all $j \neq j'$,
\begin{align*}
\cayley{\gamma_j \alpha_j \alpha'_j \beta_j}{\gamma_{j'} \alpha_{j'} \alpha'_{j'} \beta_{j'}} & > \Kc'(\dist(w_j,w_{j'})+1) + 2N\base + 2\base \\
& \geq \Kc(\dist(w_j,w_{j'}) + 2N + 1) \\
& \geq \Kc(\dist(w_jw_j',w_{j'}w_{j'}')+1)
\end{align*}
\qed\end{proof}

\section{Proofs of Section~\ref{sec:result}: Main result}
\label{appendix_result}

We now prove that the branching twinning property of order~$k$ implies the $k$-sequentiality (Proposition \ref{prop:BTP-kseq}). It is developed in the rest of this Section.


\subsection*{Proof of Proposition~\ref{prop:BTP-kseq}}

We present here the elements missing from the sketch of proof.
First, we expose the subset construction with delays, together with the proofs of its properties, adapted from \cite{BealC02} to fit our settings.
Then, we present in details the proof of the properties \prop{1} and \prop{2}.
Finally, we give the formal construction of the sequential weighted automata $\overline{V}_i$, $1 \leq i \leq k$, whose union is equivalent to $W$.

\subsubsection*{Subset construction with delays.}

Given $W = (Q, t_{\textsl{init}}, t_{\textsl{final}}, T)$ a weighted automaton, we construct an equivalent infinite sequential automaton $D_W = (Q', t'_{\textsl{init}}, t'_{\textsl{final}}, T')$ as follows.
The states of $D_W$ are the subsets of $Q \times \group$.
For every $S \in Q'$ and $a \in A$, we define the single transition starting from $S$, and labelled by the input $a$, as follows.
\begin{enumerate}
\item
First, the elements of $S$ are updated with respect to $W$.
Let
\[
S' =  \{(q,\beta \gamma) | \exists p \in Q \textup{ s.t. }(p,\beta) \in S, (p,a,\gamma,q) \in T\} \subseteq Q \times \group.
\]
\item
Then a reference pair $(q,\alpha) \in S'$ is picked, and $S'$ is normalized with respect to it.
Let $(p,\alpha) = \selec{S'}  \in Q \times \group$, where $\sele$ is a choice function.
Let
\[
S'' := \{(q,\alpha^{-1} \beta) | (q,\beta) \in S'\}  \subseteq Q \times \group.
\]
Note that, in particular, $(q,\unit) \in S''$.
\end{enumerate}
Then $(S,a,\alpha,S'') \in T'$.

For every state $S \in Q'$, let $W_S$ denote the weighted automaton obtained by replacing the initial output relation of $W$ with $S$, i.e., $W_S = (Q,S,t_{final},T)$.
Note that $W_{t_{\textsl{init}}} = W$.
Before defining the initial relation and the final relation of $D_W$, let us prove two lemmas that follow from the definition of $T'$.

\begin{lemma}\label{lemma_norm}
For every accessible state $S$ of $D_W$ that is not the initial state, there exists $q \in Q$ such that $(q,\unit) \in S$.
\end{lemma}

\begin{proof}
This follows immediately from the definition of $T'$.
\end{proof}

\begin{lemma}\label{lemma_witness}
Let $S \subset Q \times \group$ and let $u \in A^*$, and consider the run $S \xrightarrow{u|\alpha} S'$ of $D$.
Then
 \[
 S' = \{ (q,\beta) | \textup{ there exists a path } \rho: \xrightarrow{\gamma} p \xrightarrow{u|\delta} q \textup{ in } W_S  \textup{ s.t. } \alpha \beta = \gamma \delta \}.
 \]
\end{lemma}

\begin{proof}

The proof is done by induction over the length of the input word $u$.
If $u = \epsilon$, then $S = S'$, $\alpha = \unit$, and the result follows immediately.
Now suppose that $u = va$ for some $v \in A^*$ and $a \in A$, and that the lemma is true for the input word $v$.
Note that the initial run can be split as follows.
\[
S \xrightarrow{v|\alpha_1} S_1 \xrightarrow{a|\alpha_2} S'.
\]
By induction hypothesis,
 \[
 S_1 = \{ (q_1,\beta_1) | \textup{ there exists a path } \rho: \xrightarrow{\gamma} p \xrightarrow{v|\delta_1} q_1 \textup{ in } W_S  \textup{ s.t. } \alpha_1 \beta_1 = \gamma \delta_1 \}.
 \]
 Moreover, by definition of the transition relation of $D_W$, 
 \[
 S' =  \{(q,\alpha_2^{-1} \beta_1 \gamma_2) | \exists q_1 \in Q \textup{ s.t. }(q_1,\beta_1) \in S_1, (q_1,a,\gamma_2,q) \in T\} \subseteq Q \times \group.
 \]
 The desired result follows, since we obtain, by applying those two equalities,
 \[
 \begin{array}{lll}
 S' & = &  \{(q,\alpha_2^{-1} \beta_1 \gamma_2) | \exists q_1 \in Q \textup{ s.t. }(p_1,\beta_1) \in S_1, (q_1,a,\gamma_2,q) \in T\}\\
 & = &  \{(q,\alpha_2^{-1} \beta_1 \gamma_2) | \exists \rho: \xrightarrow{\gamma} p \xrightarrow{v|\delta_1} q_1 \textup{ in } W_S  \textup{ s.t. } \alpha_1 \beta_1 = \gamma \delta_1 \textup{ and } (q_1,a,\gamma_2,q) \in T\}\\
  & = &  \{(q,\alpha_2^{-1} \beta_1 \gamma_2) | \exists \rho: \xrightarrow{\gamma} p \xrightarrow{u|\delta} q \textup{ in } W_S  \textup{ s.t. } \alpha_1 \beta_1 \gamma_2 = \gamma \delta \}\\
 & = &  \{(q,\beta) | \exists \rho: \xrightarrow{\gamma} p \xrightarrow{u|\delta} q \textup{ in } W_S  \textup{ s.t. } \alpha_1 \alpha_2 \beta = \gamma \delta\}\\
 & = &  \{(q,\beta) | \exists \rho: \xrightarrow{\gamma} p \xrightarrow{u|\delta} q \textup{ in } W_S  \textup{ s.t. } \alpha \beta = \gamma \delta\}.
 \end{array}
 \]
 \end{proof}

We can now define the initial relation and the final relation of $D_W$.
Since Lemma \ref{lemma_witness} guarantees that, starting from state $S$, $D_W$ accurately simulates the weighted automaton $W_S$, and $W_{t_{\textsl{init}}} = W$,  by setting $t'_{\textsl{init}} = (t_{\textsl{init}},\unit)$, and by setting 
\[
t'_{\textsl{final}} = \{ (S,\alpha \beta) | \textup{ there exists } q \in Q \textup{ s.t. } (q,\alpha) \in S \textup{ and } (q,\beta) \in t_{\textsl{final}} \},
\]
we ensure that $D_W$ is equivalent to $W$.

%

\subsubsection*{Proof of \prop{1}}

We now demonstrate \prop{1}, using Lemma \ref{lemma_witness} and the fact that $W$ is $\ell$-valued.

\begin{lemma} \label{lemma_count}
Let $S$ be an accessible state of $D_W$.
Then $|S| \leq \ell |Q|$. 
\end{lemma}

\begin{proof}
For every $q \in Q$, let $n_q \in \mathbb{N}$ denote the number of times that $q$ appears in $S$, i.e., $n_q = |\{ \beta \in \group | (q,\beta) \in S \}|$.
Using the fact that $W$ is $\ell$-valued, we now prove that $n_q \leq \ell$, which implies the desired result, since $n = \Sigma_{q \in Q}n_q$.

Let $\beta_1, \ldots, \beta_{n_q}$ be an enumeration of the elements of $\{ \beta \in \group | (q,\beta) \in S \}$.
Since $S$ is accessible, there exists a run $t_{\textsl{init}} \xrightarrow{u|\alpha} S$ in $D_W$, and, by Lemma \ref{lemma_witness}, for every $0 \leq j \leq n_q$, since $(q,\beta_j) \in S$, there exists an initial run
\[
\rho_j: \xrightarrow{\gamma_j} p_j \xrightarrow{u|\delta_j} q
\]
of $W_{t_{\textsl{init}}} = W$ such that $\alpha \beta_j = \gamma_j \delta_j$.
Moreover, since $W$ is trim by supposition, it admits a final run
\[
\rho: q \xrightarrow{v|\delta} q_f \xrightarrow{\delta_f}.
\]
Then, for every $0 \leq j \leq n_q$, $(uv,\alpha \beta_j \delta \delta_f) = (uv,\gamma_j \delta_j \delta \delta_f) \in \inter{C}$.
Since the $\beta_j$ are distinct elements of the group $\group$, so are the elements $\alpha \beta_j \delta \delta_f \in \group$, hence, since $W$ is $\ell$-valued by supposition, $n_q \leq \ell$.
\end{proof}

\subsubsection*{Proof of \prop{2}}

In order to prove \prop{2}, we first introduce a new definition.

\begin{definition}
For every state $S \in Q'$, let the \textit{rank} of $S$, denoted by $\rank{S}$, be the minimal integer $k'$ such that $W_S$ satisfies the $\BTP{k'}$, where $W_S$ is the weighted automaton obtained by replacing the initial output relation of $W$ with $S$, i.e., $W_S = (Q,S,t_{final},T)$.
\end{definition}

\begin{lemma}\label{lemma_split}
Let $N_W = 2M_W|Q|^{\ell |Q|}$.
Let $S$ be an accessible state of $D_W$ that contains a pair $(q, \alpha)$ such that $|\alpha|>N_W$.
Then there exists a partition of $S$ into two subsets $S'$ and $S''$ such that $\rank{S'} + \rank{S''} \leq k$.
\end{lemma}

\begin{proof}
By Lemma \ref{lemma_norm}, we know that there exists a pair $(q_1,\unit) \in S$. 
Moreover, by supposition, there exists a pair $(q_2,\alpha_2) \in S$ such that $|\alpha_2| > N_W$.

Let $(q_3,\alpha_3), \ldots, (q_m,\alpha_m)$ be an enumeration of the elements
of $S$ distinct from $(q_1,\unit)$ and $(q_2,\alpha_2)$.
By Lemma \ref{lemma_count}, $m \leq \ell |Q|$.
Since $S$ is accessible, there exists a run $t_{\textsl{init}} \xrightarrow{u|\alpha} S$ in $D_W$.
Then, by Lemma \ref{lemma_witness}, for every $1 \leq j \leq m$, since $(q_j, \alpha_j) \in S$, there exist a run
\[
\rho_j: \xrightarrow{\gamma_{0.j}} q_{0.j} \xrightarrow{u|\gamma_{j}}  q_{j}
\]
of $W$ over $u$ such that $\alpha \alpha_j = \gamma_{0.j} \gamma_{j}$.

The proof is now done in two steps.
First, we expose a decomposition of $u$ into three words $wv_{\inds}w'$ such that every run $\rho_j$ loops over $v_{\inds}$, and the delay between the outputs of the runs $\rho_1$ and $\rho_2$ is modified along $v_{\inds}$.
This allows us to define a partition $\{S', S''\}$ of $S$, splitting the elements $(q_j,\alpha_j)$ of $S$ depending on whether or not the delay between $\rho_1$ and $\rho_j$ changes along $v_{\inds}$.
Then, we prove that $\rank{S'} + \rank{S''} \leq k$, using the following idea.
Let $r' = \rank{S'}-1$, $r'' = \rank{S''}-1$.
By combining a witness of the non satisfaction of the $\BTP{r'}$ by $W_{S'}$ and a witness of the non satisfaction of the $\BTP{r''}$ by $W_{S''}$, we build a witness of the non satisfaction of the $\BTP{r' + r'' + 1}$ by $W_{S}$.
This implies that
\[
\rank{S'} + \rank{S''} - 1  = r' + r'' + 1 < k,
\]
and the desired result follows.

\begin{enumerate}
\item
By applying Lemma \ref{lem-split-run} to the product of $m$ copies of $W$, we obtain a subdivision $u_0v_1u_1 \dotsm v_{|Q|^m}u_{|Q|^m}$ of the word $u$ such that each run $\rho_j$ loops over each input $v_{s}$, and $|u_0 \dotsm u_{|Q|^m}| < |Q|^{m} \leq |Q|^{\ell |Q|}$.

Note that the distance between the outputs of $\rho_1$ and $\rho_2$ after reading the input $u$ is
\[
\cayley{\alpha_1}{\alpha_2} = \cayley{\unit}{\alpha_2}  > N_W = 2M_W|Q|^{\ell |Q|}.
\]
In order for the distance to increase by at least $N_W$, there has to exist an integer $1 \leq \inds \leq m$ such that the delay between the outputs of $\rho_1$ and $\rho_2$ changes along $v_{\inds}$, since $M_W$ is greater than or equal to the maximal size of the output of a transition of $W$, and $|u_0 \dotsm u_{m}| < |Q|^{\ell |Q|}$.
Let $w = u_0v_1u_1 \dotsm v_{\inds-1}u_{\inds-1}$, $w' = u_{\inds}v_{\inds+1}u_{\inds+1} \dotsm v_{|Q|^{m}}u_{|Q|^{m}}$, and for every $1 \leq j \leq m$, consider the following decomposition of the run $\rho_j$:
\[
\rho_j : \xrightarrow{\gamma_{0,j}} q_{0,j} \xrightarrow{w|\gamma_{1,j}} p_j \xrightarrow{v_{\inds}|\gamma_{2,j}} p_j \xrightarrow{w'|\gamma_{3,j}} q_j.
\]
Let $S'$ be the set of pairs $(q_j,\alpha_j)$ corresponding to the indices $j$ such that the delay between $\rho_1$ and $\rho_j$ stays the same along $v_{\inds}$, i.e.,
\[
\delay{\gamma_{0,1} \gamma_{1,1}}{\gamma_{0,j} \gamma_{1,j}} = \delay{\gamma_{0,1} \gamma_{1,1}\gamma_{2,1}}{\gamma_{0,j} \gamma_{1,j}\gamma_{2,j}},
\]
Then, let $S'' = S \setminus S'$.
By definition of $v_{\inds}$, $S''$ is not empty since it contains $(q_2,\alpha_2)$, hence $\{S', S''\}$ is a partition of $S$.

\item

\newcommand{\state}{p}
\newcommand{\inpr}{w}
\newcommand{\outr}{\delta}
\newcommand{\inpl}{x}
\newcommand{\outl}{\zeta}

Let $r' = \rank{S'}-1$, $r'' = \rank{S''}-1$, and $r = r' + r''$.
By definition of the rank, there exists an unsatisfied instance of the \BTP{r'} over the weighted automaton $W_{S'}$.
Let $\phi_0, \ldots, \phi_{r'}$ be the runs of $W_{S'}$ forming this instance.
Similarly, there exists an unsatisfied instance of the \BTP{r''} over the weighted automaton $W_{S''}$.
Let $\phi_{r' + 1}, \ldots, \phi_{r + 1}$ be the runs of $W_{S''}$ forming this instance.
For every $1 \leq \indr \leq r+1$, we add loops over the empty word at the end of the run $\phi_{\indr}$ in order for it to contain exactly $r$ loops, yielding the run
\begin{tikzpicture}[->,>=stealth',shorten >=1pt,auto,node distance=2.8cm, scale=0.8]
    \tikzstyle{every state}=[fill=purple!10,text=black,minimum width=1cm,scale=0.9]
    \tikzstyle{every edge}=[draw=black,font=\small]
    \tikzstyle{initial}=[initial by arrow, initial where=left, initial text=]

    \node (I) at (-3.25,1.5) {$\phi'_{\indr}$ : };
    \node (A0) at (-1,1.5) {$\state_{0,\indr}$};
    \node (A1)  at (2.25,1.5) {$\state_{1,\indr}$};
    \node (A)  at (5.75,1.5) {$\dotsm$};
    \node (A3)  at (9.25,1.5) {$\state_{r',\indr}.$};

    \path (I) edge 
    	 node {$\outr_{0,\indr}$} (A0);
    \path (A0) edge 
    	 node {$\inpr_{1,\indr} \mid \outr_{1,\indr}$} (A1);
    \path (A1) edge 
    	 node {$\inpr_{2,\indr} \mid \outr_{2,\indr}$} (A);
    \path (A) edge 
    	 node {$\inpr_{r,\indr} \mid \outr_{r,\indr}$} (A3);

    \path (A1) edge [loop above] node {$\inpl_{1,\indr} \mid \outl_{1,\indr}$} (A1);
    \path (A3) edge [loop above] node {$\inpl_{r,\indr} \mid \outl_{r,\indr}$} (A3);

  \end{tikzpicture}

Note that, since both $S'$ and $S''$ are subsets of $S$, each $\phi'_{\indr}$ can be seen as a run over the weighted automaton $W_S$.
By extending those runs on the left, we now construct an instance of the $\BTP{r+1}$ for $W$ that is not satisfied.
For every $0 \leq \indr \leq r$,  $(\state_{0,\indr}, \outr_{0,\indr}) \in S$, therefore there exists $1 \leq j_{\indr} \leq m$ such that $(\state_{0,\indr}, \outr_{0,\indr}) = (q_{j_{\indr}},\alpha_{j_{\indr}})$, hence we can compose the run $\rho_{j_{\indr}}$ with $\phi_{\indr}$, as follows.

\begin{tikzpicture}[->,>=stealth',shorten >=1pt,auto,node distance=2.8cm, scale=0.8]
    \tikzstyle{every state}=[fill=purple!10,text=black,minimum width=1cm,scale=0.9]
    \tikzstyle{every edge}=[draw=black,font=\small]
    \tikzstyle{initial}=[initial by arrow, initial where=left, initial text=]

    \node (I) at (-3.25,1.5) {$\psi_{\indr}$ : };
    
    \node (A0) at (-1,1.5) {$q_{0,{j_{\indr}}}$};
    \node (A1)  at (2.25,1.5) {$p_{j_{\indr}}$};
    \node (A2)  at (5.5,1.5) {$\state_{1,\indr}$};
    \node (A)  at (7.75,1.5) {$\dotsm$};
    \node (A3)  at (10,1.5) {$\state_{r,\indr}$,};

    \path (I) edge 
    	 node {$\gamma_{0,{j_{\indr}}}$} (A0);
    \path (A0) edge 
    	 node {$w \mid \gamma_{1,{j_{\indr}}}$} (A1);
    \path (A1) edge 
    	 node {$w_{\indr} \mid \gamma_{\indr}$} (A2);
    \path (A2) edge 
    	 node {$\inpr_{2,\indr} \mid \outr_{2,\indr}$} (A);
    \path (A) edge 
    	 node {$\inpr_{r,\indr} \mid \outr_{r,\indr}$} (A3);

    \path (A1) edge [loop above] node {$v_{\inds} \mid \gamma_{2,{j_{\indr}}}$} (A1);
    \path (A2) edge [loop above] node {$\inpl_{1,\indr} \mid \outl_{1,\indr}$} (A2);
    \path (A3) edge [loop above] node {$\inpl_{r,\indr} \mid \outl_{r,\indr}$} (A3);

  \end{tikzpicture}

where $w_{\indr} = v_{\inds} w' \inpr_{1,\indr}$ and $\gamma_{\indr} = \gamma_{2,j_{\indr}} \gamma_{3,j_{\indr}} \outr_{1,\indr}$.
Note that $\psi_{\indr}$ is a run of $W$.
In order to conclude the proof, we need to prove that the instance of the $\BTP{r+1}$ for $W$ formed by the runs $\psi_0, \ldots, \psi_r$ is not satisfied.
For every $0 \leq \indr < \indr' \leq r+1$, we now expose a loop differentiating $\psi_{\indr}$ and $\psi_{\indr'}$, i.e., a loop along which the delay between the outputs of $\psi_{\indr}$ and $\psi_{\indr'}$ changes, and that occurs on the part of the runs where the inputs are identical.
We consider three possibilities.
\begin{itemize}
\item
If $0 \leq \indr \leq r' < \indr' \leq r+1$, then $(q_{j_{\indr}},\alpha_{j_{\indr}}) \in S'$ and $(q_{j_{\indr'}},\alpha_{j_{\indr'}}) \in S''$, and, by definition of the partition $\{S',S''\}$ of $S$,
\[
\delay{\gamma_{0.{j_{\indr}}} \gamma_{1.{j_{\indr}}}}{\gamma_{0.{j_{\indr'}}} \gamma_{1.{j_{\indr'}}}} \neq \delay{\gamma_{0.{j_{\indr}}} \gamma_{1.{j_{\indr}}}\gamma_{2.{j_{\indr}}}}{\gamma_{0.j_{{\indr'}}} \gamma_{1.{j_{\indr'}}}\gamma_{2.{j_{\indr'}}}},
\]
Therefore, the first loop of $\psi_{\indr}$ and $\psi_{\indr'}$ can be used to differentiate them.
\item
If $0 \leq \indr < \indr' \leq r'$, since the runs $\phi_0, \ldots, \phi_{r'}$ form a non satisfied instance of the $\BTP{r'}$ for $W_{S'}$, we can find a loop differentiating $\phi_{\indr}$ and $\phi_{\indr'}$, and use the corresponding loop to differentiate $\psi_{\indr}$ and $\psi_{\indr'}$.
\item
If $r' + 1 \leq \indr < \indr' \leq r+1$, since the runs $\phi_{r' +1}, \ldots, \phi_{r+1}$ form a non satisfied instance of the $\BTP{r''}$ for $W_{S''}$, we can find a loop differentiating $\phi_{\indr}$ and $\phi_{\indr'}$, and use the corresponding loop to differentiate $\psi_{\indr}$ and $\psi_{\indr'}$.
\end{itemize}
\end{enumerate}
\end{proof}

Finally, \prop{2} is obtained as a corollary.

\begin{corollary}
Let $S$ be an accessible state of $D_W$ that contains a pair $(q, \alpha)$ such that $|\alpha|>N_W$.
Then $W_S$ is $k$-sequential.
\end{corollary}

\begin{proof}
By the previous lemma, there exists a partition of $S$ into two subsets $S'$ and $S''$ such that $\rank{S'} + \rank{S''} \leq k$.
Note that $\rank{S'} \geq 1$ and $\rank{S''} \geq 1$, hence $\rank{S''} < k$ and $\rank{S'}<k$
Therefore, the induction hypothesis can be applied, proving that $W_{S'}$ is $\rank{S'}$-sequential, and $W_{S''}$ is $\rank{S''}$-sequential.
Finally, as $S$ is equal to the union of $S'$ and $S''$, $W_S$ is equivalent to the union of $W_{S'}$ and $W_{S''}$.
Therefore $W_S$ is $k$-sequential.
\end{proof}

\subsubsection*{Final construction }
We construct $k$ sequential weighted automata $\overline{V}_1, \ldots, \overline{V}_k$ whose union is equivalent to $W$. 
Let $U$ denote the set containing the accessible states $S$ of $D_W$ that contain only pairs $(q,\alpha)$ satisfying $|\alpha| \leq N_W$.
Moreover, let $U'$ be the set of states of $D_W$ accessible in one step from $U$, i.e.,
\[
U' = \{ S' | \exists S \in U, a \in A, \alpha \in \group \textup{ s.t. } (S,a,\alpha,S') \in T' \}.
\]
As there are only finitely many $\alpha \in \group$ such that $|\alpha| \leq N_W$, \prop{1} implies that $U$ is finite.
Note that this implies the finiteness of $U'$.
By $\prop{2}$, for every state $S \in U'$ that is not in $U$, $W_{S}$ can be expressed as the union of $k$ sequential weighted automata $V_i(S)$, with $1 \leq i \leq k$.
For every $1 \leq i \leq k$, let $\overline{V}_i$ be defined as the union of $D_W$ restricted to the states $S \in U$, and all the $V_i(S')$, for $S' \in U' \setminus U$, with the two following differences.
First, the only initial state of $\overline{V}_i$ is the initial state of $D_W$.
Second, for every transition $(S,a,\alpha,S')$ of $D_W$ between states $S \in U$ and $S' \in U' \setminus U$, we add a transition $(S,a, \alpha \alpha_{S'.i},q_{S'.i})$, where $\{(q_{S'.i},\alpha_{S'.i})\}$ is the initial relation of $V_{i}(S')$.

\begin{proposition}
The weighted automaton $W$ is equivalent to the union of the $\overline{V}_i$, $1 \leq i \leq k$.
\end{proposition}

\begin{proof}
We prove that $D_W$ is equivalent to the union of the $\overline{V}_i$, $1 \leq i \leq k$, which implies the desired result, since $W$ is equivalent to $D_W$.

First, we show that the relation defined by $D_W$ is included into the union of the $\inter{\overline{V}_i}$.
Let $\rho: \xrightarrow{} S_0 \xrightarrow{u|\alpha} S_f \xrightarrow{\beta}$ be an accepting run of $D_W$.
We now expose an integer $i \in \{ 1, \ldots, k \}$ such that the output of the run of $\overline{V}_i$ over the input $u$ is $\alpha \beta$.
If all the states visited by $\rho$ are in $U$, $\rho$ is present in each $\overline{V}_i$, and we are done.
Otherwise, let us split $\rho$ as follows.
\[
\rho: \xrightarrow{} S_0 \xrightarrow{u_1| \alpha_1} S \xrightarrow{a|\alpha'} S' \xrightarrow{u_2|\alpha_2} S_f \xrightarrow{\beta},
\] 
where $S'$ is the first state encountered along $\rho$ that is not in $U$.
Then $S' \in U'$.
Moreover, by Lemma \ref{lemma_witness}, the definition of the final relation of $D_W$, and the fact that $W_{S'}$ is equivalent to the union of the $V_i(S')$, there exists $1 \leq i \leq k$ and a run 
\[
\rho': \xrightarrow{\gamma_0} q_0 \xrightarrow{u_2 | \gamma_1} q_f \xrightarrow{\gamma_2}
\]
in $V_i(S')$ such that $\gamma_0 \gamma_1 \gamma_2 = \alpha_2 \beta$.
Then the run
\[
\xrightarrow{} S_0 \xrightarrow{u_1| \alpha_1} S \xrightarrow{a|\alpha' \gamma_0} q_0 \xrightarrow{u_2|\gamma_1} q_f \xrightarrow{\gamma_2}
\]
over the input $u$ is in $\overline{V}_i$, and the associated output is $\alpha_1 \alpha' \gamma_0 \gamma_1 \gamma_2 = \alpha \beta$, which proves the desired result.

Conversely, we can prove, using similar arguments, that for every $1 \leq i \leq k$, $\inter{\overline{V}_i}$ is included into $\inter{D_W}$, which concludes the proof.
\end{proof}

\section{Proofs of Section~\ref{sec:cra}: Cost register automata with independent registers}
\label{appendix_cra}
%

\paragraph*{Proof of Proposition~\ref{prop:kseq-crak-equivalence}.}


\begin{proof}
  From a $k$-sequential weighted automaton $\cup_{i \in \{ 1, \ldots, k \}} W_i$, we can build a CRA $C_i$ with $1$ register for each of the $W_i$. We obtain a CRA with $k$ independent registers by making the product of those $C_i$.
  From a CRA $C$ with $k$ independent registers $\{ X_i \}_{i \in \{ 1, \ldots, k \}}$, for each $i$ we can produce a trim projection of $C$ on the register $X_i$. Each of these $1$-register machines can be expressed as a weighted automaton, and their union is a $k$-sequential weighted automaton.
\qed\end{proof}

\section{Proofs of Section~\ref{sec:transducers}: The case of transducers}
\label{appendix_trans}
%

\paragraph*{Notations.}

Let $\monoid=(M,\otimes,\unit)$ be a monoid.
Given $O,O'\subseteq M$, $O \otimes O'$
(or simply $OO'$) is the set $\{\alpha \beta \mid \alpha\in
O,\ \beta\in O'\}$, $O^k$ denotes the set $\underbrace{O O\dotsm
O}_{k \text{ times}}$, $O^{<k}=\cup_{0 \leq i < k} O^i$ and $O^{\leq
k} = O^{<k}\cup O^k$.

Moreover, if $\monoid$ is a group and $O\subseteq M$, $O^{-1}$ denotes the set
$\{\alpha^{-1} \mid \alpha \in O\}$.

\paragraph*{Proof of Proposition \ref{proposition:positivite}.}

Let us sketch the main ideas.
  Consider a cost register automaton $\reg$ over $\free$ that computes a relation in
  $A^*\times B^*$. One can prove that there is a bound $N'$ such that along
  the runs of $\reg$, the values stored in the registers always belong
  to $\alphab^*(\alphab \cup \alphab^{-1})^{\leq N'}$.
 This
  intuitively relies on the fact that for every run that
  can be completed into
  an accepting run, there exists a ``short'' completion, and this completion
  should lead to a weight in $\alphab^*$.
  At anytime during a computation, the values stored in registers
are thus of the form $\alpha_1\alpha_2$ with $\alpha_1
  \in \alphab^*$ and $\alpha_2 \in (\alphab \cup \alphab^{-1})^{\leq
    N'}$.
 For a given register $X$, the idea is then to associate with $X$
 the shortest $\alpha_1$ satisfying
  these conditions, and to store the value $\alpha_2$ in the states of the
  automaton. This ensures that every continuation of the computation
  will be compatible with the value $\alpha_1$ already computed. We use
  here the fact that the weights are elements of the free group in order
  to prove the existence of a ``shortest'' $\alpha_1$.
  This construction preserves the fact that $C$ uses $k$ independent registers.

Let us give a formal proof now.
Let $(Q, q_{init}, \Reg, \delta, \mu)$ denote a register automaton over $\free = (B\cup B^{-1})^*$ computing a relation $F \subseteq A^* \times B^*$.
The Cayley distance is defined in $(\free, B)$.

Let $N=|Q|m+s$ where:
\begin{itemize}
\item $m$ is the maximum of the set:
\begin{align*}
\{|\alpha| \mid & \delta(q,a)=(p,h), h(Y)=(X,\alpha), q,p\in Q, a\in \alphabet, X, Y \in \Reg \}
\end{align*}
\item $s$ is the maximum of the set:
$$\{|\alpha| \mid (q,X,\alpha) \in \mu, X \in \Reg, q\in Q\}$$
\end{itemize}

We construct now an equivalent register automaton $\reg'$ over $\alphab^*$.
Set $\mathcal{G}$ the set of functions $\Reg \to (\alphab \cup \alphab^{-1})^{\leq N}$.

The set of states of $\reg'$ is $Q' = Q\times \mathcal{G}$. The initial state is $(q_{init},r_{init})$ where $r_{init}$ is the function that associates each register to $\unit$ and the set of registers is $\Reg$.

For $\alpha \in \alphab^* (\alphab \cup \alphab^{-1})^{\leq N}$, let us denote by $\alpha_1$ and $\alpha_2$ the two elements such that $\alpha=\alpha_1 \alpha_2$ and $\alpha_1$ is the shortest word in $\alphab^*$ such that $\alpha_2 \in (\alphab \cup \alphab^{-1})^{\leq N}$. Remark that $\alpha_1$ and $\alpha_2$ always exist in this case.

The transition function $\delta'$ is defined in the following way: given $q\in Q$, $a\in \alphabet$, let $(p,g) = \delta(q,a)$. We set:
$\delta'((q,r),a)=((p,t),h)$ where $h(Y)=(X,(r(X)\alpha)_1)$, $t(Y)=(r(X)\alpha)_2$ for $X,Y \in \Reg$ such that $g(Y)=(X,\alpha)$, if $r(X)\alpha \in \alphab^* (\alphab \cup \alphab^{-1})^{\leq N}$. If $r(X)\alpha \in \alphab^* (\alphab \cup \alphab^{-1})^{\leq N}$ (we will see that it is never the case for accessible, co-accessible states), then we set: $\delta'((q,r),a)=((p,t),h)$ where $h(Y)=(X,r(X)\alpha)$, $t(Y)=\unit$ for $X,Y \in \Reg$ such that $g(Y)=(X,\alpha)$.

The output relation $\mu'$ is defined by the triplets $((q,r),X,r(X)\alpha)$ for all $(q,X,\alpha) \in \mu$.

\bigskip

First, it is easy to check that $\reg$ and $\reg'$ compute the same relation. It relies on the following lemma:

\begin{lemma}
\label{lemma:induction}
For all words $w$, there is a run in $\reg$ on $w$ from $(q_{init},\nu_{init})$ to some $(q,\nu)$ if and only if there is a run in $\reg'$ on $w$ from $((q_{init},r_{init}),\nu_{init})$ to $((q,r),\sigma)$ such that for all registers $X$, $\nu(X)=\sigma(X)r(X)$.
\end{lemma}

\begin{proof}
The proof is made by induction on the length of $w$. By construction the property holds for $w = \varepsilon$. Suppose now that $w=w'a$ for some $a \in \alphabet$.

Suppose that there is a run in $\reg$ on $w$ from $(q_{init},\nu_{init})$ to $(q,\nu)$. Set $(q',\nu')$ the configuration such that there is a run in $\reg$ on $w'$ from $(q_{init},\nu_{init})$ to $(q',\nu')$ and $\delta(q',a)=(q,g)$. By induction hypothesis, there is a run in $\reg'$ on $w'$ from $((q_{init},r_{init}),\nu_{init})$ to $((q',r'),\sigma')$ such that for all registers $X$, $\nu'(X)=\sigma'(X)r'(X)$. Moreover, by construction, we are in one of the following case:
\begin{itemize}
\item $\delta'((q',r'),a)=((q,r),h)$ where $h(Y)=(X,(r'(X)\alpha)_1)$, $r(Y)=(r'(X)\alpha)_2$ for $X,Y \in R$ such that $g(Y)=(X,\alpha)$ if $r(X)\alpha \in \alphab^* (\alphab \cup \alphab^{-1})^{\leq N}$. Thus there is a run in $\reg$ on $w$ from $((q_{init},r_{init}),\nu_{init})$ to $((q,r),\sigma)$ with $\sigma(Y)r(Y) = \sigma'(X)(r'(X)\alpha)_1(r'(X)\alpha)_2$ for $X,Y \in \Reg$ such that $g(Y)=(X,\alpha)$. Thus, $\sigma(Y)r(Y) = \nu'(X)\alpha = \nu(Y)$.
\item  $\delta'((q',r'),a)=((p,t),h)$ where $h(Y)=(X,r'(X)\alpha)$, $r(Y)=\unit$ if $r(X)\alpha \notin \alphab^* (\alphab \cup \alphab^{-1})^{\leq N}$. Thus there is a run in $\reg$ on $w$ from $((q_{init},r_{init}),\nu_{init})$ to $((q,r),\sigma)$ with $\sigma(Y)r(Y) = \sigma'(X)(r'(X)\alpha)$ for $X,Y \in \Reg$ such that $g(Y)=(X,\alpha)$. Thus, $\sigma(Y)r(Y) = \nu'(X)\alpha = \nu(Y)$.
\end{itemize}

Conversely, suppose that there is a run in $\reg'$ on $w$ from $((q_{init},r_{init}),\nu_{init})$ to $((q,r),\sigma)$. Set $((q',r'),\sigma')$ the configuration such that there is a run in $\reg'$ on $w'$ from $((q_{init},r_{init}),\nu_{init})$ to $((q',r'),\sigma')$ and $\delta'((q',r'),a)=((q,r),h)$. Then by induction hypothesis, there is a run in $\reg$ on $w'$ from $(q_{init},\nu_{init})$ to $(q',\nu')$ such that for all registers $X$, $\nu'(X)=\sigma'(X)r'(X)$. Moreover, by construction,
$\delta(q,a)=(q,g)$ and if $g(Y)=(X,\alpha)$ then, suppose that we are in the first case, $h(Y)=(X,(r'(X)\alpha)_1)$, $r(Y)=(r'(X)\alpha)_2$ for $X,Y \in \Reg$. Thus, there is a run in $\reg$ on $w$ from $(q_{init},\nu_{init})$ to $(q,\nu)$ with $\nu(Y) = \nu'(X)\alpha = \sigma'(X)r'(X) \alpha = \sigma'(X)(r'(X)\alpha)_1(r'(X)\alpha)_2 = \sigma(Y)r(Y)$. The second case is similar.
\qed\end{proof}

By the previous lemma, we can now prove that $\inter{\reg} = \inter{\reg'}$.
Consider a run in $\reg$ on a word $w$ from $(q_{init},\nu_{init})$ to $(q,\nu)$ and $(q,Y,\alpha) \in \mu$. Then, by Lemma \ref{lemma:induction}, there is a run in $\reg'$ on $w$ from $((q_{init},r_{init}),\nu_{init})$ to $((q,r),\sigma)$ for some $\sigma$ such that $\nu(Y)=\sigma(Y)r(Y)$. Thus, $\nu(Y)\alpha = \sigma(Y)r(Y)\alpha$ and by construction, $((q,r),Y,r(Y)\alpha) \in \mu'$.

Conversely, suppose that there is a run in $\reg'$ on $w$ from $((q_{init},r_{init}),\nu_{init})$ to $((q,r),\sigma)$ and a register $Y$ such that $((q,r),Y,r(Y)\alpha) \in \mu'$. Then by Lemma \ref{lemma:induction}, there is a run in $\reg$ on $w$ from $(q_{init},\nu_{init})$ to $(q,\nu)$ such that $\nu(Y)=\sigma(Y)r(Y)$. Thus, $\sigma(Y)r(Y)\alpha = \nu(Y)\alpha$ and by construction, $(q,Y,\alpha) \in \mu$.

Thus, $\reg'$ computes the same relation as $\reg$. Moreover, by construction, $\reg'$ uses $k$ independent registers.
What is left is to prove that it only uses elements in $\alphab^*$.

\bigskip

Set $E$ the minimal subset of $Q\times \Reg$ such that:
\begin{itemize}
\item $(q,X) \in E$ if there is $\alpha\in \free$ such that $(q,X,\alpha) \in \mu$,
\item $(q,X) \in E$ if there is $(p,Y) \in E$, $a\in \alphabet$ such that $\delta(q,a)=(p,h)$ for some $h$ such that $h(Y)=(X,\alpha)$ for some $\alpha \in \free$.
\end{itemize}

We say that a register $X$ is \intro{alive} in $q$ if $(q,X) \in E$.

\begin{lemma}
\label{lemma:trans}
For all configurations $(q,\nu)$, for all runs from $(q_{init}, \nu_{init})$ to $(q,\nu)$, for all alive registers $X$ in $q$, $\nu(X)$ belongs to $\alphab^*(\alphab \cup \alphab^{-1})^{\leq N}$.
\end{lemma}

\begin{proof}
Consider a run from $(q_{init}, \nu_{init})$ to $(q,\nu)$ and an alive register $X$ in $q$ such that $\nu(X)=c \in \free$.
The run can be completed in a run that ends in some $(q',\nu')$ such that there are a register $Y$, $\alpha \in (\alphab \cup \alphab^{-1})^{\leq |Q|m}$, $\beta \in (\alphab \cup \alphab^{-1})^{\leq s}$ with $\nu'(Y) = c \alpha$, $(q',Y,\beta) \in \mu$ and thus $c \alpha \beta \in \alphab^*$.
Finally, $c \in \alphab^* (\alphab \cup \alphab^{-1})^{\leq N}$.
\qed\end{proof}

To prove that the updates of $\reg'$ only use elements in $\alphab^*$, we need to prove that for all accessible $(q,r)$, for all $X,Y \in \Reg$, $X$ alive in $q$, $\alpha \in \group$ such that $g(Y)=(X,\alpha)$, we have $r(X)\alpha$ belongs to $\alphab^* (\alphab \cup \alphab^{-1})^{\leq N}$. We prove it by induction on the length of the shortest run ending in $(q,r)$. It is true for $(q_{init},r_{init})$ by definition of $N$.

By contradiction, if it is not true for $(q,r)$, then $r(X) \alpha = u a^{-1} v$ with $u\in \alphab^*$, $a \in \alphab$, $v\neq av'$ for all $v'\in \free$, and $v \notin (\alphab \cup \alphab^{-1})^{< N}$. By induction hypothesis and completing the run, there is a word $w$ such that $(w,u' u a^{-1} v v') \in \inter{\reg}$ with $u' \in \alphab^*$ and $v' \in (\alphab \cup \alphab^{-1})^{\leq N}$. This output is supposed to be in $\alphab^*$.
We are in one of the two following cases:
\begin{itemize}
\item The word $u$ ends with a letter of $\alphab$ different from $a$. In this case, $u' u a^{-1} v v'$ cannot be in $\alphab^*$ since $v\neq av'$ for all $v'\in \free$, $v \notin (\alphab \cup \alphab^{-1})^{< N}$ and $v' \in (\alphab \cup \alphab^{-1})^{\leq N}$. That is a contradiction.
\item The word $u$ is empty. Since $v\neq av'$ for all $v'\in \free$, $v \notin (\alphab \cup \alphab^{-1})^{< N}$ and $v' \in (\alphab \cup \alphab^{-1})^{\leq N}$, then the last letter of $u'$ must be an $a$. By definition of $\delta'$, if this $a$ has been used to update the registers, it means that in $\reg$, the run corresponding to this output has a sequence of weights: $a, \alpha_1,\ldots, \alpha_s, a^{-1}$ such that $\alpha_1 \dotsm \alpha_s = \unit$, $\alpha_1 \dotsm \alpha_s \neq a^{-1}\beta$ for all $\beta\in \group$ and by definition of $\delta'$, $\alpha_1 \dotsm \alpha_s = \beta \gamma$ with $\beta \notin B^*(\alphab \cup \alphab^{-1})^{\leq N}$. Under these conditions, $u'\beta \notin B^*(\alphab \cup \alphab^{-1})^{\leq N}$, that is in contradiction with Lemma~\ref{lemma:trans}. $(\star)$
\end{itemize}

\medskip

Finally, we need to prove that the output relation of $\reg'$ also uses only elements in $\alphab^*$. Thus, let us prove that for all accessible $(q,r)$ and $(q,X,\alpha) \in \mu$, $r(X)\alpha$ belongs to $\alphab^*$. By contradiction, if not, then $r(X)\alpha = u a^{-1} v$ with $a\in \alphab$, $u \in \alphab^*$ that does not end with $a$, $v \in \free, v\neq av'$ for some $v' \in \free$. Let us treat the two following cases:
\begin{itemize}
\item The word $u$ ends with a letter $b\neq a$. In this case, by completing the run, one of the images of a word is of the form $\beta r(X) \alpha$ with $\beta \in \alphab^*$. But in this case, $\beta r(X) \alpha$ would not belong to $\alphab^*$. Thus, $\reg$ would not compute a relation in $A^*\times B^*$, that is a contradiction.
\item The word $u$ is empty. By completing the run into an accepting run in $\reg'$, we get a sequence of weights of transitions $\alpha_1, \ldots, \alpha_s \in \alphab^*$ such that one of the output is $\alpha_1\dotsm \alpha_s a^{-1} v$. Since the output is supposed to be in $\alphab^*$, it means that the word $\alpha_1\dotsm \alpha_s$ ends with an $a$. And we can use a similar argument as $(\star)$.
\end{itemize}

\section{Proofs of Section~\ref{sec:decision}: Decidability}
\label{appendix_decision}

First observe that when we described the skeleton of a counter-example,
we claimed that we can look for a counter-example with more than $k$
loops. This claim is exactly Lemma~\ref{lemma-BTP-loops},
proven in Appendix, Section~\ref{appendix_loops_BTP}.

\paragraph{Case of commutative groups.}
We omitted the proof of the lower bound, so we give details now.

\begin{lemma}
Over the group $(\mathbb{Z},+)$, the \BTP{k} problem is \PSPACE-hard~($k$ given in unary).
\end{lemma}
\begin{proof}
We present a reduction of the emptiness of $k$ deterministic finite state automata
to the \BTP{k} problem, using similar ideas to a lower bound proved in~\cite{DBLP:conf/icalp/AlurR13}.
Let $D_1,\ldots,D_k$ be $k$ deterministic finite state automata over some alphabet $A$.
Let $\#$ and $\$$ be two fresh symbols not in $A$. For each $i$,
we build a deterministic weighted automaton $W_i$ over the group $(\mathbb{Z},+)$.
The semantics of $W_i$ is defined as:
$$
\inter{W_i} = \{ (u\# \$^m, i*m) \mid u \in \dom(D_i), m \in \mathbb{N}\}
$$
We consider now the weighted automaton defined as $W = \bigcup_i W_i$.
We claim that $W$ does not satisfy  $\BTP{k-1}$ iff the intersection
of the languages defined by the $D_i$'s is non-empty.

First, suppose that the intersection of the languages is empty. We describe the construction
of a CRA with $k-1$ independent registers $X_1,\ldots,X_{k-1}$ that realizes the relation $\inter{W}$.
We consider the deterministic finite state automaton $D$ obtained as the product of the $D_i$'s.
We add to $D$ the set of states $Q' = 2^{\{1,\ldots,k\}}$.
From each state $(q_1,\ldots, q_k)$ of $D$, we add a transition on symbol $\#$ that goes to
the state $\{i \mid q_i \textup{ is a final state of }D_i\} \in Q'$.
As the intersection of the languages is empty, every reachable state $q'\in Q'$ is such that
$|q'|\leq k-1$. We can thus restrict $D$ to elements of $Q'$ of size at most $k-1$.
Last, given a state $q'\in Q'$, we add a self-loop labelled by the symbol $\$$.

We describe now the updates of the registers. The only non-trivial updates are on self-loops of
states in $Q'$. Let $q'=\{i_1,\ldots,i_m\} \in Q'$. As explained above, we have $m\leq k-1$.
We suppose that indices $i_j$'s are sorted by increasing order. Intuitively, we use register
$X_j$ to represent index $i_j$, so we define the update $X_j:=X_j*i_j$. The final output function
of state $q'$ is simply defined as the set $\{X_1,\ldots,X_m\}$.
It is easy to verify that this CRA realizes the relation $\inter{W}$.

Conversely, suppose now that the intersection of the languages is non-empty, and let $u$
be a word in this intersection. We proceed by contradiction, and suppose that
there exists an equivalent CRA with $k-1$ independent registers.
By definition, the relation $\{(u\#\$^m, i*m) \mid m \geq 0, i\in \{1,\dots, k\}\}$
is included in $\inter{W}$. In particular we can find input words with $k$ output values that are pairwise
arbitrarily far. It is easy to show that a CRA with $k-1$ registers is unable to accept such a relation, hence
a contradiction.
\qed\end{proof}

\paragraph{Case of transducers.}

We describe the condition on the skeleton that we can ensure when we are considering transducers.

Let $\Kc$ be a positive integer.
We say that two runs $\rho_1$ and $\rho_2$ on the same input word $u$ are $\Kc$-close if
for every prefix $u'$ of $u$, the restrictions $\rho'_1$ and $\rho'_2$ of the two runs on the
input $u'$ are such that $\cayley{\alpha_1}{\alpha_2}\leq \Kc$, where $\alpha_i$ is the output word produced
by the run $\rho'_i$.

\begin{lemma}[Nice shape]
Let $T$ be a transducer, then $T$ violates $\BTP{k}$ iff there exists a counter example given as follows: there are
\begin{itemize}
\item  states $\{q_{i,j}\}_{0\leq i \leq m, 0\leq j \leq k}$ with $k \leq m \leq k^2$, and $q_{0,j}$ initial for all $j$,
\item  words  $\{u_{i,j}\}_{0\leq i \leq m, 0\leq j \leq k}$ and $\{v_{i,j}\}_{1\leq i \leq m, 0\leq j \leq k}$ such that there are $k+1$ runs satisfying for all $0 \le j \le k$, for all $1 \le i \le m$, $q_{i-1,j} \xrightarrow{u_{i,j} \mid \alpha_{i,j}} q_{i,j}$ and $q_{i,j} \xrightarrow{v_{i,j} \mid \beta_{i,j}} q_{i,j}$
\end{itemize}
and such that for all $0\leq j <j' \leq k$, there exists $1\leq i \leq m$ such that for all $1\leq i' \leq i$, we have $u_{i',j}=u_{i',j'}$ and $v_{i',j}=v_{i',j'}$, and
\begin{enumerate}[$a)$]
\item either  $|\beta_{i,j}|\neq |\beta_{i,j'}|$,
\item or  $|\beta_{i,j}| = |\beta_{i,j'}| \neq 0$,
the words $\alpha_{1,j}\ldots \alpha_{i,j}$ and $\alpha_{1,j'}\ldots \alpha_{i,j'}$ have a
mismatch, and the runs $q_{0,j} \xrightarrow{u_1\ldots u_i}q_{i,j}$
and $q_{0,j'} \xrightarrow{u_1\ldots u_i}q_{i,j'}$ are $M_T.n^{k+1}$-close.
\end{enumerate}
\end{lemma}

\begin{proof}[Sketch]
The reverse implication is trivial, so we focus on the direct one. We consider a counter-example to
the $\BTP{k}$ and aim at deriving a counter example satisfying the above properties.

Let us consider a pair $(j,j')$ of runs indices with $0\leq j <j' \leq k$.
By definition, there exists an index $i$ (satisfying $i\leq \chi(j,j')$) such that the loop $i$ induces a different delay.
If this is due to the length of the output words, then we are done as we are in case $a)$.

Otherwise, let us assume that every loop of index at most $\chi(j,j')$ has output words of same length on components
$j$ and $j'$. Consider a loop that induces different delays. Then it induces a mismatch between the non-empty output words
of the loop on components $j$ and $j'$. This loop can be unfolded to move the mismatch on output words of the runs leading to the loop.
It remains to show that the runs are $M_T.n^{k+1}$-close. If this is the case, then we are done
and have proven that case $b)$ is satisfied.
Otherwise, we can prove that the input word  has length at least $n^{k+1}$ and thus that there exists
a synchronized loop (on the $k$ runs) whose output words on components
$j$ and $j'$ have distinct length. But then we are back in case $a)$.
\qed\end{proof}

\fi

\end{document}